\documentclass[a4paper,11pt]{article}

\usepackage{amsthm}
\usepackage{mathptmx}
\usepackage{amssymb}
\usepackage{amsfonts}
\usepackage{tikz}
\usepackage{authblk}
\usepackage{amsmath,amscd}
\usepackage{pgfplots}
\usepackage{geometry}
\usepgflibrary{shapes.geometric} 
\usepgflibrary[shapes.geometric] 
\usetikzlibrary{shapes.geometric} 
\usetikzlibrary{arrows,shapes,backgrounds}
\usepackage{pifont}
\usepackage{multicol}        
\usepackage{amsmath,amsthm,wasysym,txfonts,indentfirst,fancyhdr,amscd, multirow}
\usepackage[T1]{fontenc}
\usepackage[utf8]{inputenc}
\RequirePackage{graphicx}

\newtheorem{teo}{Theorem}[section]
\newtheorem{lemma}{Lemma}[section]
\newtheorem{defn}{Definition}[section]
\newtheorem{alg}{Algorithm}[section]
\newtheorem{prop}{Proposition}[section]

\newcommand{\Z}{\mathbb{Z}}
\newcommand{\I}{\mathcal{I}}

\newcommand{\R}{\mathbb{R}}

\newcommand{\D}{\mathcal{D}}
\newcommand{\La}{\mathcal{L}}
\newcommand{\Po}{\mathcal{P}}
\newcommand{\G}{\mathcal{G}}
\newcommand{\A}{\mathcal{A}}
\newcommand{\h}{\mathcal{H}}

\newcommand{\Or}{\mathcal{O}}

\newcommand{\Aut}{\operatorname{Aut}}
\newcommand{\e}{\mathrm{e}}

\title{Orbits of crystallographic embedding of non-crystallographic groups and applications to virology} 
\date{}
\author[1,2,3]{Reidun Twarock}
\author[4]{Motiejus Valiunas}
\author[1,3]{Emilio Zappa \thanks{ez537@york.ac.uk}}
\affil[1]{Department of Mathematics,}
\affil[2]{Department of Biology,}
\affil[3]{York Center for Complex Systems Analysis, University of York}
\affil[4]{Faculty of Mathematics, University of Cambridge, UK} 

\begin{document}
\maketitle

\begin{abstract}
The architecture of infinite structures with non-crystallographic symmetries can be modeled via aperiodic tilings, but a systematic construction method for finite structures with non-crystallographic symmetry at different radial levels is still lacking. We present here a group theoretical method for the construction of finite nested point set with non-crystallographic symmetry.  Akin to the construction of quasicrystals, we embed a non-crystallographic group $G$ into the point group $\Po$ of a higher dimensional lattice and construct the chains of all $G$-containing subgroups. We determine the orbits of lattice points under such subgroups, and show that their projection into a lower dimensional $G$-invariant subspace consists of nested point sets with $G$-symmetry at each radial level. The number of different radial levels is bounded by the index of $G$ in the subgroup of $\Po$. In the case of icosahedral symmetry, we determine all subgroup chains explicitly and illustrate that these point sets in projection provide blueprints that approximate the organisation of simple viral capsids, encoding information on the structural organisation of capsid proteins and the genomic material collectively, based on two case studies. Contrary to the affine extensions previously introduced, these orbits endow virus architecture with an underlying finite group structure, which lends itself better for the modelling of its dynamic properties than its infinite dimensional counterpart.   
\end{abstract}

\section{Introduction}

Non-crystallographic symmetries are ubiquitous in physics, chemistry and biology. Prominent examples are quasicrystals, alloys exhibiting five-, eight-, ten- and twelve-fold symmetry with long-range order in their atomic organisation \cite{steurer} and, in carbon chemistry,  icosahedral  carbon cage structures called fullerenes  \cite{kroto1}, with architectures akin to Buckminster Fuller's geodesic domes. Icosahedral symmetry also plays a fundamental role in virology. Viruses encapsulate and hence protect their genomic material inside a protein shell, called capsid, that in the vast majority of cases possesses icosahedral symmetry.  In 1962, Caspar and Klug proposed, in their seminal paper \cite{caspar}, a theory to describe the geometry of icosahedral viral capsids and predict the locations and orientations of the capsid proteins. Inspired by the structure of the geodesic dome, they derived a series of icosahedral triangulations, called deltahedra. More recently, \cite{reidun} proposed a generalisation of this theory by considering more general tilings of the capsid surface inspired by the theory of quasicrystals. 

Caspar-Klug theory and generalisations thereof descibe the capsid of a virus as a two-dimensional object rather than in three-dimensional space. Therefore, they do not provide information about other important features of the capsid, such as its thickness and the organisation of the genomic material encapsulated inside. Experiments showed that many viruses exhibit icosahedral symmetry at different radial levels: examples are the dodecahedral cage of RNA observed in Pariacoto Virus \cite{tang} and the double-shell structure of the genomic material of Bacteriophage MS2 \cite{toropova}. These results suggest that the symmetry of the virus should be extended to include information on the capsid proteins and the packaged genome collectively.

A first step towards this goal was the principle of affine extensions, described in a series of papers \cite{keef, tom, jess}. In this work, the generators of the icosahedral group have been extended by a non-compact generator acting as a translation, with the additional requirement that the resulting words of the group satisfy non-trivial relations. Such affine extension can also be obtained via a construction similar to affine extensions in the theory of Kac-Moody algebras \cite{carter}. In this case, icosahedral symmetry is extended via an extension of the Cartan matrix, resulting in the addition of a non-compact operator to the generators of the icosahedral group \cite{patera1, dechant2, dechant3}. The orbits of the affine extensions thus constructed consist of infinite sets of points that densely fill the space, since the icosahedral group is non-crystallographic in 3D. Since viral capsids are finite objects, a cut-off parameter must be introduced, that limits the number of monomials of the affine group. In previous work, words characterised by a finite action of the translation operator had been used to construct multi-shell structures, in which each radial level displays icosahedral symmetry. However, such a cut-off implies that the point sets are not invariant under the extended group structure, which limits the use of these concepts in the formulation of energy functions, e.g. Hamiltonians or in the context of Ginzburg-Landau theory.

An alternative to this approach based on affine extension is Janner's work, that models viral architecture in terms of lattices. In a series of papers \cite{janner, janner1, janner2, janner3, janner4}, Janner embedded virus structure into lattices, and showed that this provided approximation for virus architecture, and provides an alternative approach for the modeling of the onion-like fullerenes \cite{jannerB}, including a paper combining this lattice approach with the affine extensions mentioned above \cite{jannerC}. Subsequently, approximations of virus architecture in terms of quasilattices were developed \cite{david}, which provide an alternative to the lattice approach by Janner, and by construction have vertex sets that contain the point arrays determined by the affine extended groups as subsets.  All of these approaches approximate viruses in terms of infinite structures, lattices, quasilattices, infinite groups, which require a cut-off. In the case of the lattices and quasilattices, the cut-off consists of choosing a subset of infinite (quasi)lattice, and in the case of affine extensions the action of the translation operator has to be limited. This motivates the study of the present paper in which we develop an approach in terms of mathematical concepts that are intrinsically finite-dimensional, because they are related to orbits of finite groups.    

In this paper we introduce a new group theoretical method to study nested point sets with non-crystallographic symmetries, based on the embedding of non-crystallographic groups into higher dimensional lattices \cite{senechal}. This embedding is a standard way to define mathematically quasicrystals, e.g. via the cut-and-project schemes and model sets \cite{moody}, or the superspace approach \cite{superspace}. More generally, in order to model objects in 3D that possess a non-crystallographic symmetry at different radial levels, it makes sense to embed the non-crystallographic symmetry into a crystallographic setting and use the long-range order implied to induce in projection information on the collective arrangements of different radial levels. \cite{jannerA} gave a first approach in this direction, by analysing double-shell structures with five-fold symmetry as projected orbits of specific point groups in higher dimensions.  Here we present a more systematic study for general non-crystallographic symmetries.  Specifically, we embed a non-crystallographic group $G$ into the point group $\Po$ of a lattice in the minimal higher dimension where the cut-and-project construction is possible. Since this embedding is not, in general, maximal, we consider the subgroups $K$ of $\Po$ containing $G$ as a subgroup, which extend the symmetry described by $G$. We prove that the projection of the orbits of lattice points under such subgroups into a lower dimensional subspace invariant under $G$ is a nested finite point set with non-crystallographic symmetry $G$. We show that the number of distinct radial levels in the projected orbits is bounded by the index of $G$ in $K$. 

As an illustration of this approach, we provide analytically an explicit construction of planar nested structures for non-crystallographic dihedral groups. Moreover, in order to pave the way for applications to icosahedral viruses, we apply this approach to the icosahedral group $\I$, which can be embedded into the point group of the simple cubic lattice in six dimensions \cite{zappa}. We classify all the $\I$-containing subgroups of the hyperoctahedral group, with the aid of the software \cite{GAP}, which is designed for problems in computational group theory. Since the 6D lattice is infinite, a cut-off parameter must be introduced in order to select a finite number of lattice points whose orbits can then be used to model the capsid. By construction, all point arrays have full icosahedral symmetry, i.e. containing reflections as well as rotations. Since viruses are known to be chiral, this may seem perplexing; however, we note that point arrays do not fully constrain viral architecture, thus proteins can be positioned in the capsid so as to break the full symmetry, as long as they adhere overall to the blueprint indicated by the points. Therefore, it is not possible to obtain a full classification of the orbits as was done by  \cite{jess}. However, these results provide for the first time a finite group structure, albeit in a higher dimensional space, underlying the geometry of the multiple layers of material in a virus. This has important consequences for the modelling of physical properties; specifically, conformational changes of viral capsids, which are important for the virus to become infective, can be modeled in the framework of the Ginzburg-Landau theory of phase transitions \cite{zappa2}, via the formulation of an energy function invariant under the generators of the symmetry group of the capsid.  

The paper is organised as follows. After reviewing, in Section \ref{embedding}, the embedding of non-crystallographic groups into higher dimensional lattices, in Section \ref{orbits} we describe the new group theoretical setup to model finite nested structures with non-crystallographic symmetry.  As a first application, we study in Section \ref{dihedral_groups} planar nested point sets obtained from projection of extensions of embedded non-crystallographic dihedral groups. In Section \ref{ico_section} we analyse in detail the case of icosahedral symmetry, classifying the chain of subgroups containing the icosahedral group embedded into the six-dimensional hyperoctahedral group. Finally, in Section \ref{virus} we use these results to obtain geometric constraints on viral capsid architecture, and present two case studies, namely the capsids of Pariacoto Virus and Bacteriophage MS2, whose structures have been intensively studied experimentally. We conclude in Section \ref{conclusion} by discussing further applications of these results.

\section{Crystallographic embedding of non-crystallographic groups}\label{embedding}

Our new group theoretical setup relies on the embedding of non-crystallographic symmetries into the point group of higher dimensional lattices. This is a standard method in the theory of quasicrystals; here we briefly review it and fix the notation that we are going to use throughout the paper. We refer to \cite{senechal} and \cite{grimm} for further information. 

The point group $\Po$ of a lattice $\La$ in $\R^d$ with generator matrix $B$ is the maximal set of orthogonal transformations that leave the lattice invariant:
\begin{equation}\label{point_group}
\Po(\La) := \{ Q \in O(d) : \exists M \in GL(d,\R) : QB = BM \}.
\end{equation}

$\Po$ is a finite group and does not depend on the matrix $B$ \cite{senechal}. The lattice group $\Lambda$ constitutes an integral representation of the point group $\Po$ with respect to $B$:
\begin{equation}\label{lattice_group}
\Lambda(B) := \{ M \in GL(d,\Z) : \exists Q \in \Po(\La) : B^{-1}QB = M \}.
\end{equation}

A finite group of isometries $G$ is non-crystallographic in dimension $k$ if it does not leave any lattice invariant in $\R^k$. Following \cite{levitov}, we introduce the following:
\begin{defn}\label{cryst_def}
Let $G \subseteq O(k)$ be a finite non-crystallographic group of isometries. A \emph{crystallographic representation} of $G$ is a matrix group $\widetilde{G}$ satisfying the following conditions:
\begin{enumerate}
\item[\emph{($\mathbf{C1}$)}] $\widetilde{G}$ stabilises a lattice $\La$ in $\R^d$, with $d > k$, i.e. $\widetilde{G}$ is a subgroup of the point group $\Po$ of $\La$;
\item[\emph{($\mathbf{C2}$)}] $\widetilde{G}$ is reducible in $GL(d,\R)$ and contains an irreducible representation $\rho_k$ of $G$ of degree $k$, i.e. 
\begin{equation}\label{cryst}  
\widetilde{G} \simeq \rho_k \oplus \rho', \qquad \text{deg}(\rho') = d-k.
\end{equation}
\end{enumerate} 
\end{defn}

The condition ($\mathbf{C1}$) implies that the matrices representing the elements of $\widetilde{G}$ with respect to a generator matrix $B$ of the lattice are integral or, equivalently, $B^{-1}\widetilde{G}B$ is a subgroup of the lattice group $\Lambda \subseteq GL(d,\Z)$ of $\La$  (cf. \eqref{lattice_group}).  As a consequence, the character $\chi_{\widetilde{G}}$ is an integer-valued function. The condition ($\mathbf{C2}$) is necessary for the construction of quasicrystals in $\R^k$ via the cut-and-project method (\cite{senechal}, \cite{grimm}).  

The minimal dimension $d > k$ for which a crystallographic representation $\widetilde{G}$ of $G$ is possible is called the \emph{minimal crystallographic dimension} of $G$. The conditions $\chi_{\widetilde{G}} \in \Z$ and \eqref{cryst} can be easily verified with the aid of the character table of $G$ and Maschke's Theorem \cite{fulton} . The existence, and possibly an explicit construction, of lattices in $\R^d$ whose point group contains a crystallographic representation of $G$ is a more difficult task. In the case of icosahedral symmetry, the minimal crystallographic dimension is six and the lattices in $\R^6$ have been classified in \cite{levitov} (this is explained in more detail in Section \ref{ico_section}). For planar non-crystallographic symmetries described by the dihedral groups $\D_{2n}$, the minimal crystallographic dimension is $\varphi(n)$, the Euler function of $n$. We will go back to this example in Section \ref{dihedral_groups}. 

Let us denote by $V^{(k)}$ the invariant subspace of $\R^d$ which carries the irrep $\rho_k$ of $G$. Let $\pi : \R^d \rightarrow V^{(k)}$ be the projection into $V^{(k)}$, i.e. the linear operator such that the diagramme
\begin{equation}\label{diagramma}
\begin{CD}
\R^d @>\widetilde{G}(g)>> \R^d\\
@VV\pi V @VV\pi V\\
V^{(k)} @> \rho_k(g) >> V^{(k)}
\end{CD}
\end{equation}
commutes for all $g \in G$:
\begin{equation}\label{comm}
\pi (\widetilde{G}(g) v) = \rho_k(g) (\pi (v)), \qquad \forall g \in G, \; \forall v \in \R^d.
\end{equation}

Let $V^{(d-k)}$ denote the orthogonal complement of $V^{(k)}$ in $\R^d$. We recall the following Proposition (for the proof, see \cite{senechal}):
\begin{prop}\label{one_to_one}
The following are equivalent:
\begin{enumerate}
\item  $V^{(d-k)}$ is \emph{totally irrational}, i.e. $V^{(d-k)} \cap \La = \{ 0 \}$;
\item $\pi \mid_{\La}$ is one to one. 
\end{enumerate}
\end{prop}

The triple $(V^{(k)}, V^{(d-k)}, \La)$ is the starting point to define model sets via cut-and-project schemes with $G$-symmetry \cite{moody}, which is a standard way to define quasicrystals mathematically. In this paper, we construct finite point sets resulting from the projection into $V^{(k)}$ of orbits of points of $\La$ under $G$-containing subgroups of the point group $\Po$. We explain this construction in the next section.  

\section{Nested point sets obtained from projection}\label{orbits}

The embedding of a non-crystallographic group $G$ into a higher dimensional lattice $\La$ is, in general, not maximal. This means that there exist proper subgroups of the point group $\Po$ of $\La$ which contain a crystallographic representation $\widetilde{G}$ of $G$ as a subgroup. Therefore, we introduce the set:
\begin{equation}\label{subgroups}
\A_{\widetilde{G}} := \{ K \leq \Po : \widetilde{G} \leq K \},
\end{equation}
which consists of all the $\widetilde{G}$-containing subgroups of $\Po$. For computational purposes, we fix the generator matrix $B$ of $\La$, and consider the subgroup structure of the lattice group $\Lambda$ in that representation, i.e. the set $\A_{\h}(B) = \{ K < \Lambda : \h < K \}$, with $\h := B^{-1} \widetilde{G} B$. Notice that a different choice in the generator matrix of the lattice results in subgroups $K'$ conjugate to $K$ in $GL(d, \Z)$.   

The elements in $\A_{\h}$ encode the symmetry described by $G$ plus additional generators that extend this symmetry. Let $K$ be an element of $\A_{\h}$, and let $n := [K : \h]$ be the index of $\h$ in $K$. Let $T = \{g_1, \ldots, g_n \}$ be a transversal of $\h$ in $K$, i.e. a system of representatives in $K$ of the right cosets of $\h$ in $K$ \cite{holt}. Let $v \in \La$ be a lattice point, which can be written as $v = (m_1, \ldots, m_d)$, with $m_i \in \Z$ (since we fixed the basis $\mathcal{B}$). $v$ can be taken as seed point for the orbit $\mathcal{O}_K(v) = \{ kv : k \in K \}$ under $K$. With this setup, we prove the following theorem.

\begin{teo}\label{prop1}
Let $\Or_i(v) \equiv \Or_{\h g_i}(v) = \{ h g_i v : h \in \h \}$ be the orbit of $v \in \La$ with respect to the coset $\h g_i$, and let us denote by $P_i(v) := \pi(\Or_i(v))$ the orbit projected into $V^{(k)}$, the subspace of dimension $k$ carrying the irrep $\rho_k$ of $G$ (cf. \eqref{diagramma}). Then we have:
\begin{enumerate}
\item $P_i(v)$ is well-defined, i.e. does not depend on the choice of the transversal $T$; 
\item $P_i(v)$ retains the symmetry described by $G$;
\item $P_i(v) = P_j(v)$ if and only if 
\begin{equation}\label{condition}
g_j^{-1} \h g_i \cap \text{Stab}_K(v) \neq \varnothing .
\end{equation}
\item If $\h$ is normal in $K$, then all $P_i(v)$ have the same cardinality. 
\end{enumerate}
\end{teo}
\begin{proof}
\begin{enumerate}
\item Let $T' = (g_1', \ldots, g_n')$ be another transversal for $\h$ in $K$. This implies that there exist $\hat{h}_i \in \h$, for $i =1, \ldots, n$, such that $g_i' = \hat{h}_i g_i$. We have
\begin{equation*}
\Or'_i(v) = \Or_{\h g_i'}(v) = \{ h g_i' : h \in \h \} = \{ h \hat{h}_i g_i v : h \in \h \} = \Or_i(v), \quad i=1, \ldots, n,
\end{equation*} 
and the result follows.

\item It follows from the commutative property in \eqref{comm}; in particular, we have
\begin{equation*}
\pi(\Or_i(v)) = \{ \pi(hg_iv) : h \in \h \} = \{ \rho_k(h) \pi(g_i v) : h \in \h \} = \{ \hat{h} \pi(g_iv) : \hat{h} \in \rho_k \} =: \hat{\Or}_i(v), 
\end{equation*}
for $ i =1, \ldots, n$. The orbit $\hat{\Or}_i(v)$ has $G$-symmetry by construction. 

\item We have
\begin{align*}
& P_i(v) = P_j(v) \iff \pi(\Or_i(v)) = \pi(\Or_j(v)) \iff \Or_i(v) \underset{\text{(since $\pi$ is $1-1$)}}{=} \Or_j(v) \;  \\
& \iff \{ h g_i v ; h \in \h \} = \{ hg_j v : h \in \h \}  \iff \exists h, k \in \h : hg_i v = kg_j v \\
& \iff g_j^{-1} k^{-1} h g_i v = v \iff g_j^{-1}k^{-1} h g_i \in \text{Stab}_K(v), 
\end{align*}
which proves the statement.

\item Since $\h$ is normal in $K$, the cosets $\h g_i$ form the quotient group $K / \h$ of size $n$. Let $X := \{ \Or_i : i =1, \ldots, n \}$ be the set of all the orbits with respect to the cosets $\h g_i$. In the following, we will omit the dependence on $v$ for ease of notation. We can define an action of $K / \h$ on $X$ as $\h_i \cdot \Or_{\h_j} := \Or_{\h_i \h_j}$. This action is well defined since $K / \h$ is a group, and it is transitive since, for every element $\Or_{\h_i} \in X$, we have $\h_j \cdot \Or_{\h^{-1}_j \h_i} = \Or_{\h_i}$. Let $S_{\h} := \text{Stab}_{K / \h}(\Or_{\h})$ denote the stabiliser of $\Or_{\h}$ under this action.  With $s := |S_\h|$ we thus have by the orbit-stabiliser theorem
\begin{equation*}
r:=|X| = \frac{|K / \h|}{|S_{\h}|} = \frac{n}{s}. 
\end{equation*}

It follows that the sets $O_i(v)$ are in bijection with the left cosets of $S_{\h}$ in $K / \h$. We denote these cosets by $A_i$, for $i=1, \ldots, r$. These form a partition of the quotient group $K / \h$, which we write as
\begin{equation*}
\underbrace{\h_1^{(1)}, \ldots, \h_s^{(1)}}_{A_1}, \ldots, \underbrace{\h_1^{(i)}, \ldots, \h_s^{(i)}}_{A_i}, \ldots, \underbrace{\h_1^{(r)}, \ldots, \h_s^{(r)}}_{A_r}\,.
\end{equation*}

By construction, $\Or_{\h^{(i)}_j} = \Or_{\h^{(i)}_k}$, for $j,k = 1,\ldots, s$. Let us define the sets
\begin{equation}\label{Ki}
K^{(i)} := \bigcup_{j=1}^s \h_j^{(i)} \subseteq K, \qquad i=1, \ldots, r.
\end{equation}     

The set $\{ K^{(i)} : i=1, \ldots, r \}$ constitutes a partition of $K$, since it is the union of cosets. Moreover they all have the same order:
\begin{equation}\label{order}
|K^{(i)}| = s \cdot |\h| =: N, \qquad \forall i =1, \ldots, r.
\end{equation}

Let $S = (\h^{(1)}_1, \ldots, \h_1^{(r)})$ be a transversal for the coset of $S_{\h}$ in $K/\h$. It follows from \eqref{Ki} that $K^{(i)} = \{ k \in K : k v \in \Or_{\h_1^{(i)}} \}$, therefore 
\begin{equation*}
\Or_{K^{(i)}} = \{ k v : k \in K^{(i)} \} = \Or_{\h^{(i)}_1}\,. 
\end{equation*}

To conclude, we observe that each $K^{(i)}$ contains complete cosets of $K / \text{Stab}_K(v)$. In fact, let $k \text{Stab}_K(v)$ be a coset in $K / \text{Stab}_K(v)$. If $k \in K^{(i)}$, then an element in $k \text{Stab}_K(v)$ is of the form $k \hat{k}$, with $\hat{k} \in \text{Stab}_K(v)$, and belongs to $K^{(i)}$ since $(k \hat{k}) v = k (\hat{k}v) = k v \in \Or_{\h_1^{(i)}}$. Therefore, each $K^{(i)}$ is partitioned into $|K^{(i)}|/|\text{Stab}_K(v)|$ sets: each of these sets corresponds to a distinct point in the orbit $\Or_{\h_1^{(i)}}$. Since $|K^{(i)}| = N$ for all $i$ due to \eqref{order}, each orbit $\Or_{\h_1^{(i)}}$ has the same number of points, and hence also each $P_i(v)$, because the projection is one-to-one. 

\end{enumerate}
\end{proof}

The decomposition of $K \in \mathcal{A}_{\h}$ into cosets with respect to $\h$ induces a well-defined decomposition of the projected orbit $\pi(\Or_K(v))$ (cf. \eqref{diagramma}):
\begin{equation}\label{decomposition}
\pi(\Or_K(v)) = \bigcup_{i=1}^n \pi(\Or_i(v)) = \bigcup_{i=1}^n \Or_{\rho_k}(\pi(g_iv)).
\end{equation}
 
The point set defined by \eqref{decomposition} consists of points situated at different radial levels, since, in general, $|\pi(g_iv)| \neq |\pi(g_jv)|$ for $i \neq j$, where $|\cdot|$ denotes the standard Euclidean norm in $\R^k$ . Hence the projected orbit is an onion-like structure, with each layer being the union of the projection of orbits corresponding to different cosets. It follows that the number $r$ of distinct radial levels is bounded by the index of $\h$ in $K$. 

Using these results, we can set up a procedure to extend the non-crystallographic symmetry described by $\rho_k$ in $V^{(k)}$. In particular, let $x \in V^{(k)}$ be a seed point for the orbit of $\rho_k$. The pre-image $v = \pi^{-1}(x)$ is a point of the lattice $\La$ by construction. Let $K$ be an element of $\mathcal{A}_{\h}$. The projection of $\Or_K(v)$ contains the orbit $\Or_{\rho_k}(x)$, which corresponds to the coset $\h$ (compare with \eqref{decomposition}), and possibly more layers with $G$-symmetry.  This procedure can be iterated; let us consider the chain of subgroups in $\mathcal{A}_{\h}$:
\begin{equation*}
\h \subseteq K_1 \subseteq K_2 \subseteq \ldots \subseteq K_m \subseteq \Lambda.
\end{equation*}

By ascending the chain we obtain a chain of orbits $\Or_{K_i}(v) \subseteq \Or_{K_{i+1}}(v)$; the projection of such orbits into $V^{(k)}$ induces a chain of nested shells. We can summarise the situation in the following diagramme:

\begin{equation}\label{scheme}
\begin{CD}
\Or_{\h}(v) \subseteq \La @> \text{extend}>> \Or_{K_1}(v) @>>> \ldots @>>> \Or_{\Lambda}(v) \\
@AA\text{lift} A @VV\text{project} V @VVV @VVV\\
\Or_{\rho_k}(x)  @>>> \pi(\Or_{K_1}(v)) \supseteq \Or_{\rho_k}(x) @>>> \ldots @>>> \pi(\Or_{\Lambda}(v)) \supseteq \ldots \supseteq \Or_{\rho_k}(x) 
\end{CD}
\end{equation} 

In the next section we  present a first application of these results in the case of planar non-crystallographic symmetries. 

\section{Embedding of dihedral groups $\mathcal{D}_{2n}$ and planar nested structures}\label{dihedral_groups}

Let $n >0$ be a natural number. The dihedral group $\mathcal{D}_{2n}$ is the symmetry group of a regular $n$-gon, and consists of $n$ rotations and $n$ reflections, with presentation \cite{holt}
\begin{equation}\label{dihedral}
\D_{2n} = \langle R_n,S : R_n^n = e, SR_n = R_n^{-1}S \rangle,
\end{equation}
where $R_n$ is a rotation by $\frac{2\pi}{n}$, and $S$ a reflection. 

Let $\xi_n = \exp{\frac{2\pi i }{n}} \in \mathbb{C}$ be a primitive root of unity, and let $\Z[\xi_n]$ be the ring of integers  of the field $\mathbb{Q}(\xi_n)$. The standard embedding of $\mathcal{D}_{2n}$ into a $\phi(n)$-dimensional lattice, where $\phi(n)$ denotes the Euler function, is achieved via the \emph{Minkowski embedding} of $\Z[\xi_n]$ \cite{grimm}. Specifically, let $\G$ denote the Galois group of $\mathbb{Q}(\xi_n)$. $\G$ is isomorphic to $\Z_n^{\times} := \{ m \in \Z_n : \text{gcd}(m,n) =  1 \}$, the multiplicative group of $\Z_n$, and therefore consists of $\phi(n)$ elements. Such elements are automorphisms of $\mathbb{Q}(\xi_n)$  given by $\xi_n \mapsto \xi_n^m$, where $n$ and $m$ are coprime, and they are pairwise conjugate. We can then choose $\frac{\phi(n)}{2}$ non-conjugate elements $\sigma_i$ in $\G$, where $\sigma_1$ is the identity. The Minkowski embedding of $\Z[\xi_n]$ is then given by
\begin{equation}\label{mink_lat}
\La_{\phi(n)} := \{(x, \sigma_2(x), \ldots, \sigma_{\frac{\phi(n)}{2}}) : x \in \Z[\xi_n] \} \subseteq \mathbb{C}^{\frac{\phi(n)}{2}} \simeq \R^{\phi(n)},
\end{equation}
which is a lattice in $\R^{\phi(n)}$. The projection $\pi : \La_{\phi(n)} \rightarrow \mathbb{C}$ on the first coordinate is, by construction, one-to-one on its image $\pi(\La_{\phi(n)}) = \Z[\xi_n]$.  

We can define an action of $\D_{2n}$ in $\Z[\xi_n]$ in the following way:
\begin{equation*}
R_n \cdot x = \xi_n x, \qquad S \cdot x = \bar{x},
\end{equation*}
where $\bar{x}$ denotes the complex conjugation in $\mathbb{C}$ and $x \in \Z[\xi_n]$. Note that  this action is well-defined as every element of $\D_{2n}$ stabilises $\Z[\xi_n]$. If $g$ is an element of $I_2(n)$, $g$ can be lifted to an element $\tilde{g}$ defined by
\begin{equation}\label{lifting}
\tilde{g} \cdot \pi^{-1}(x) = \pi^{-1}(g \cdot x),
\end{equation}
which is well-defined since the projection is one-to-one. In particular, we have
\begin{equation*}
\tilde{R}_n \cdot \pi^{-1}(x) =   \pi^{-1}(R_n \cdot x) = \pi^{-1}(\xi_n x) = (\xi_nx, \sigma_2(\xi_n x), \ldots, \sigma_{\frac{\phi(n)}{2}}(\xi_n x)).
\end{equation*}

Similarly we have
\begin{equation*}
\tilde{S} \cdot \pi^{-1}(x) = \pi^{-1}(S\cdot x) = \pi^{-1}(\bar{x}) =  \left(\bar{x}, \overline{\sigma_2(x)}, \ldots, \overline{\sigma_{\frac{\phi(n)}{2}}(x)}\right).
\end{equation*}

It follows that the transformations $\tilde{R}_n$ and $\tilde{S}$ are orthogonal and stabilise the lattice $\La_{\phi(n)}$. Therefore the set $\{ \tilde{g} : g \in \D_{2n} \}$ is an embedding of $\D_{2n}$ into the point group of $\La_{\phi(n)}$. We point out that, although this construction is a priori possible and well-defined for all natural numbers, it is difficult to find the explicit form of the point group of $\La_{\phi(n)}$ in \eqref{mink_lat} for general $n$. The explicit form is known, in particular, for $n=5,8$ and $12$ \cite{grimm}.

We now prove the existence of an extension $K$ of $\D_{2n}$ embedded into $\Po(\La_{\phi(n)})$, i.e. a subgroup $K$ of $\Po(\La_{\phi(n)})$ such that $\D_{2n}$ is a normal subgroup of $K$. Note that  $\D_{2n}$ can be seen as a subgroup of the symmetric group $S_n$, acting on the vertices of a regular $n$-gon. More precisely, let $R_n' = (1,2,\ldots, n)$ be an $n$-cycle and let $S'$ be the permutation defined by $S'(j) = -j \; \text{mod}\; n$, for $j = 1,\ldots, n$; then $\langle R_n', S' \rangle$ defines a permutation representation of $\D_{2n}$.  Let $T = \langle R_n' \rangle \simeq \Z_n$, and define $K$ as the normaliser \cite{artin} of $T$:  
\begin{equation}\label{K_def}
K := N_{S_n}(T) = \{ \sigma \in S_n : \sigma^{-1} T \sigma = T \}.
\end{equation}

We point out that $K$ thus constructed is referred to as the \emph{holomorph} of the group $\Z_n$, and denoted by $\text{Hol}(\Z_n)$ \cite{hall}. We have the following:

\begin{lemma}\label{lemma_perm}
$\mathcal{D}_{2n}$ is a proper normal subgroup of $K = \text{Hol}(\Z_n)$ when $n=5$ or $n \geq 7$.
\end{lemma}
\begin{proof}
We have, using notation as in (\ref{K_def}): 
\begin{align*}
\sigma \in K & \iff \sigma T\sigma^{-1} = T \iff \sigma R_n\sigma^{-1} \in T \iff \sigma (1,2,\dots,n)\sigma^{-1} = (1,2,\dots,n)^m \\
&\text{ for some $m \in \mathbb{Z}_n$}\iff (\sigma(1),\sigma(2),\dots,\sigma(n)) = (1,2,\dots,n)^m\text{ for some $m \in \mathbb{Z}_n$}  \\
& \text{with $\text{gcd}(m,n) = 1$} \text{(otherwise $(1,2,\dots,n)^m$ decomposes into disjoint cycles)} \\
&\iff \,\forall j \in \mathbb{Z}_n, \sigma(j) = mj+l\text{ for some $m, l \in \mathbb{Z}_n$ with $\text{gcd}(m,n) = 1$.}
\end{align*}

To sum up:
\begin{equation*}
K = \{ \sigma \in S_n : \exists m \in \Z_n^{\times}, l \in \Z_n : \sigma(j) = mj+l, \forall j \in \Z_n \}.
\end{equation*}

$K$ contains $R_n'$ and $S'$, which correspond to $m = 1$, $l = 1$ and $m=-1$, $l=0$, respectively. It follows that $\D_{2n}$ is a subgroup of $K$. We notice that $|K| = \phi(n)n$, which is greater than $2n$ for $n = 5$ or $n \geq 7$. Hence $\mathcal{D}_{2n}$ is a proper subgroup of $K$ for these values of $n$, which correspond to the non-crystallographic cases.

In order to prove the normality, we write
\begin{equation*}
\D_{2n} \simeq  \langle R_n' \rangle \cup \langle R_n' \rangle S' =  T \cup T S'.
\end{equation*}

Define $\sigma \in K$ by $\sigma(j) = mj+l$. We want to prove that $\sigma \D_{2n} = \D_{2n} \sigma$. Clearly $\sigma T = T \sigma$ by definition of $K$ (cf. \eqref{K_def}). We are then left to show that $\sigma T S' = T S' \sigma$.  For $s,j \in \mathbb{Z}_n$ we have $((R_n')^s S')(j) = s-j$, which implies $\left(\sigma \left(\left(R'_n)^sS'\right)\right) \sigma^{-1}\right)(j) = ms-j+2l$. Therefore, $\sigma  \left( \left(R'_n)^sS'\right)\right) \sigma^{-1} = (R'_n)^{ms+2l}S'$, hence $\sigma T S' = T S' \sigma $, and the result follows.
\end{proof}

We now prove the following:

\begin{prop}
$K = \text{Hol}(\Z_n)$ is a subgroup of $\Po(\La_{\phi(n)})$.
\end{prop}
\begin{proof}
Let us define the functions $t_{m,l} \in \Aut(\Z[\xi_n])$ by
\begin{equation}\label{B}
t_{m,l}\left(\sum_{j=0}^{n-1} a_j \xi_n^j\right) := \sum_{j=0}^{n-1} a_j \xi_n^{mj+l}, \quad m \in \Z_n^{\times}, l \in \Z_n.
\end{equation}

Notice that the elements $t_{m,0}$, with $m \in \Z_n^{\times}$, correspond to the Galois automorphisms $\sigma_m$, which constitute the Galois group $\G$ of $\mathbb{Q}(\xi_n)$. 
Let $K' = \{ t_{m,l} : m \in \Z_n^{\times}, l \in \mathbb{Z}_n \}$. $K'$ is a $\G$-containing subgroup of $\Aut(\Z[\xi_n])$. In particular, composition of two elements is given by
\begin{equation}\label{A}
t_{m,l} \cdot t_{m',l'} = t_{mm',ml'+l},
\end{equation}
and the inverse of an element $t_{m,l}$ is $t_{m^{-1},-m^{-1}l}$.  Let $\theta : K \rightarrow K'$ be the function
\begin{equation*}
\theta(\sigma) = t_{m,l}, \qquad \sigma(j) = mj+l.
\end{equation*}

$\theta$ is an isomorphism by construction.  We define $ \mathcal{D}_{2n}' := \theta(\mathcal{D}_{2n})$.  By Lemma \ref{lemma_perm}, we have that $\mathcal{D}_{2n}'$ is a normal subgroup of $K' < \text{Aut}(\Z[\xi_n])$.

We can write the Minkowski embedding of  $\mathbb{Z}[\xi_n]$ as $\mathcal{L}_{\phi(n)} = \{ t_{y_1,0}(z), \dots, t_{y_{\phi(n)/2},0}(z) : z \in \mathbb{Z}[\xi_n] \} \in \mathbb{C}^{\frac{1}{2}\phi(n)} \cong \mathbb{R}^{\phi(n)}$, where $1 = y_1 < \cdots < y_{\phi(n)/2} < \frac{n}{2}$ and $\text{gcd}(y_j,n) = 1 $, for all $j$. We can then lift $t_{m,l}$ as in \eqref{lifting}, and obtain:

\begin{align*}
& \tilde{t}_{m,l} \cdot \pi^{-1}(z)  = \pi^{-1}(t_{m,l}(z)) =   \left(t_{m,0}(t_{m,l}(z)),\dots,t_{my_{\phi(n)/2},0}(t_{m,l}(z))\right) \\
&  \underset{(\text{by \eqref{A}})}{=} \; \left( t_{my_1,y_1l}(z), \ldots, t_{my_{\phi(n)/2}, y_{\phi(n)/2}}(z) \right) \underset{(\text{by \eqref{B}})}{=} \; \left( \xi_n^{y_1l} t_{my_1,0}(z), \ldots, \xi_n^{y_{\phi(n)/2}l} t_{my_{\phi(n)/2}}(z) \right).
\end{align*}

Therefore $\tilde{t}_{m,l}$ just rearranges the coordinates of $\pi^{-1}(z)$, possibly converting some of them to their complex conjugates and/or multipliying them by a power of $\xi_n$. Hence $K'$ stabilises the lattice $\La_n$, and this action is orthogonal. Thus $K'$ is a subgroup of $\Po(\La_n)$, and the result is proved.
\end{proof}

It follows that we can construct nested point sets with $n$-fold symmetry using the extension $K = \text{Hol}(\Z_n)$ and the Minkowski embedding $\La_{\phi(n)}$. The number $r$ of distinct radial levels obtained via projection is at most 
\begin{equation*}
r \leq [K : \h] = \frac{\phi(n) n}{2n} = \frac{\phi(n)}{2}.
\end{equation*}

\subsection{Five-fold symmetry}

As a first example, we consider the case $n = 5$. In this case, the Minkowski embedding of $\D_{10}$ is isomorphic to the root lattice $A_4$ \cite{grimm}, whose simple roots are given by $\alpha_i = e_i-e_{i+1}$, for $i = 1,\ldots, 4$, and $e_i$ denotes the standard basis of $\R^5$ (cf. also \cite{carter} for more details on root systems of semisimple Lie algebras). With respect to the basis of simple roots, we obtain a representation $\h$ of $\D_{10}$ which is a subgroup of the lattice group $\Lambda(A_4)$ (which is isomorphic to the symmetric group $S_5$):   
\begin{equation}\label{H2_rep}
\h = \left< \left( \begin{array}{cccc}
1 & 0 & 0 & 0 \\
1 & 0 & 0 & -1 \\
1 & 0 & -1 & 0 \\
1 & -1 & 0 & 0
\end{array} \right), \left( \begin{array}{cccc}
-1 & 1 & 0 & 0 \\
0 & 1 & 0 & 0 \\
0 & 1 & 0 & -1 \\ 
0 & 1 & -1 & 0
\end{array} \right) \right>.
\end{equation}

This representation splits into two two-dimensional irreps, which induce a decomposition $\R^4 \simeq E^{(1)} \oplus E^{(2)}$, where $E^{(1)}$ and $E^{(2)}$ are both totally irrational with respect to the root lattice $A_4$. A basis for each of them can be found using tools from the representation theory of finite groups \cite{fulton}. The projection $\pi_2^{(1)} : \R^4 \longrightarrow E^{(1)}$ is given by
\begin{equation}\label{H2_proj}
\pi_2^{(1)} = \frac{1}{\sqrt{2(3-\tau)}}\left( \begin{array}{cccc}
-\tau'\sqrt{3-\tau} & \sqrt{3-\tau} & 0 & -\sqrt{3-\tau} \\
-1 & 2-\tau & -2\tau' & 2-\tau 
\end{array} \right),
\end{equation}
where $\tau:= \frac{1}{2} \left( 1+\sqrt{5}\right)$ denotes the golden ratio and $\tau':= 1-\tau$ its Galois conjugate. The space $E^{(1)}$ carries the irrep $\rho_2$:
\begin{equation}\label{H2_irrep}
\rho_2 = \left< \left( \begin{array}{cc}
1 & 0 \\
0 & -1 
\end{array} \right), \frac{1}{2} \left( \begin{array}{cc}
-\tau' & \sqrt{\tau+2} \\
\sqrt{\tau+2} & \tau' 
\end{array} \right)
\right>.
\end{equation}

With the help of \texttt{GAP}, we study the set $\mathcal{A}_{\h}$ of subgroups of $\Lambda(A_4)$ containing $\h$ (compare with \eqref{subgroups}). There is a unique chain of subgroups containing a proper extension of $\h$:
\begin{equation}\label{five_fold_chain}
\h \lhd K \subseteq \Lambda(A_4),
\end{equation} 
where $K$ is, in fact, isomorphic to $\text{Hol}(\Z_5)$. The explicit representation of $K$ is given by 

\begin{equation*}
K = \left< \left( \begin{array}{cccc}
0 &  -1 &  1 &  0 \\  
-1 &  0 &  1 &  0 \\ 
 0 &  0 &  1 &  0 \\ 
 0 &  0 &  1 &  -1
\end{array} \right),  \left( \begin{array}{cccc}
0 &  0 &  0 &  -1 \\
1 &  0 &  0 &  -1 \\
0 &  1 &  0 &  -1 \\ 
0 &  0 &  1 &  -1
\end{array} \right) ,  \left( \begin{array}{cccc}
1 &  0 &  0 &  0 \\
1 &  0 &  -1 &  1 \\ 
0 &  1 &  -1 &  1 \\  
0 &  1 &  -1 &  0
\end{array}\right)
\right>. 
\end{equation*}

The group $K$ corresponds to the point group $54$ given in \cite{janner2}. We point out that, by Theorem \ref{prop1}, the point sets obtained from projection of orbits of points of the root lattice $A_4$ consists of at most two radial levels, which can either be two nested decagons of two nested pentagons. In Figure \ref{quasicrystal_fragments} we show an example of such point sets.  

\begin{figure}[!t]
\centering
\begin{minipage}[t]{0.25\textwidth}
\centering
\includegraphics[scale = 0.13]{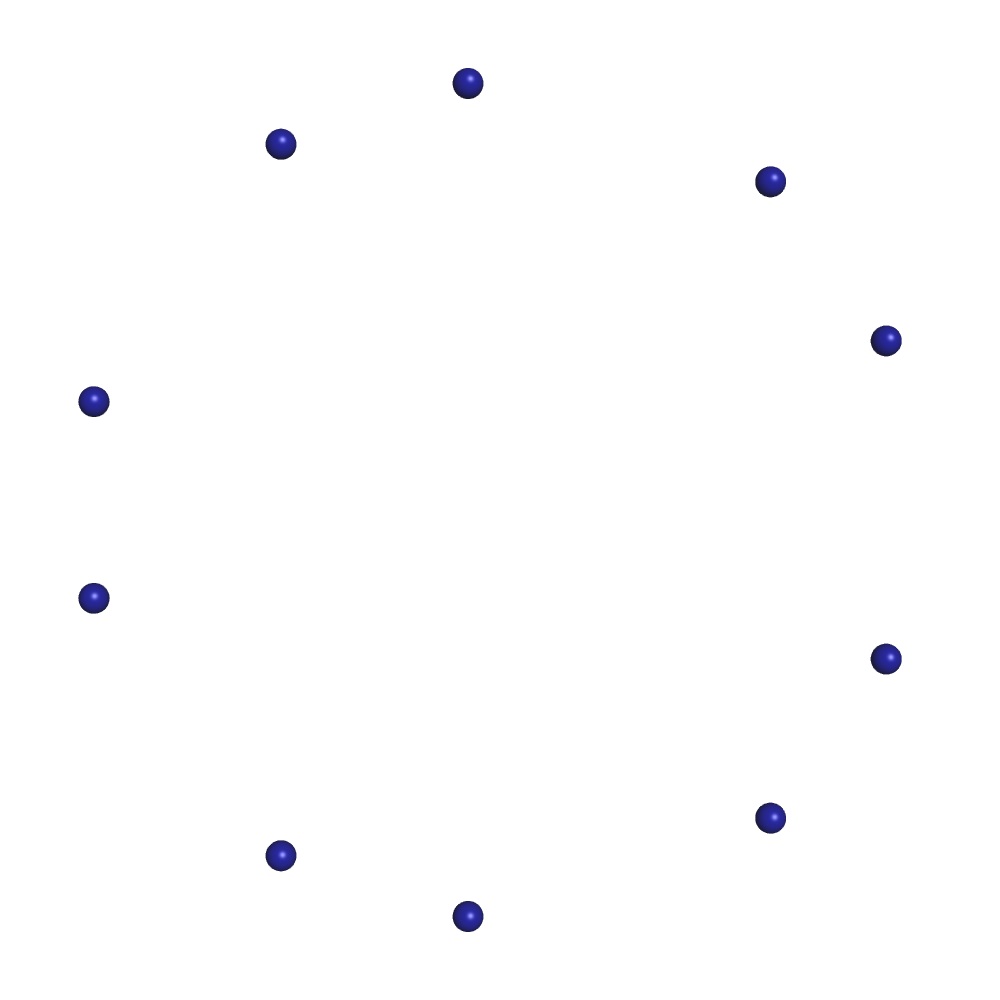}
(a)
\end{minipage}
\qquad
\begin{minipage}[t]{0.25\textwidth}
\centering
\includegraphics[scale = 0.13]{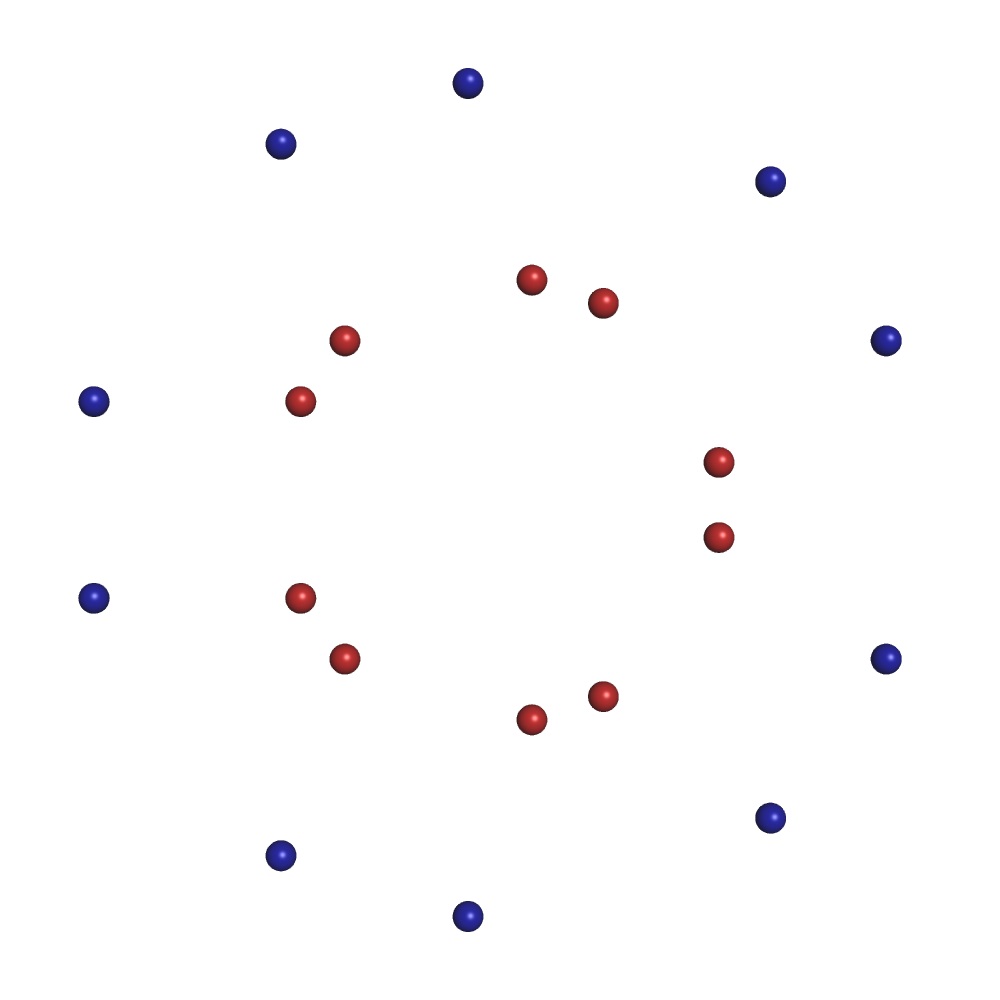} 
(b)
\end{minipage}
\qquad
\begin{minipage}[t]{0.25\textwidth}
\centering
\includegraphics[scale = 0.13]{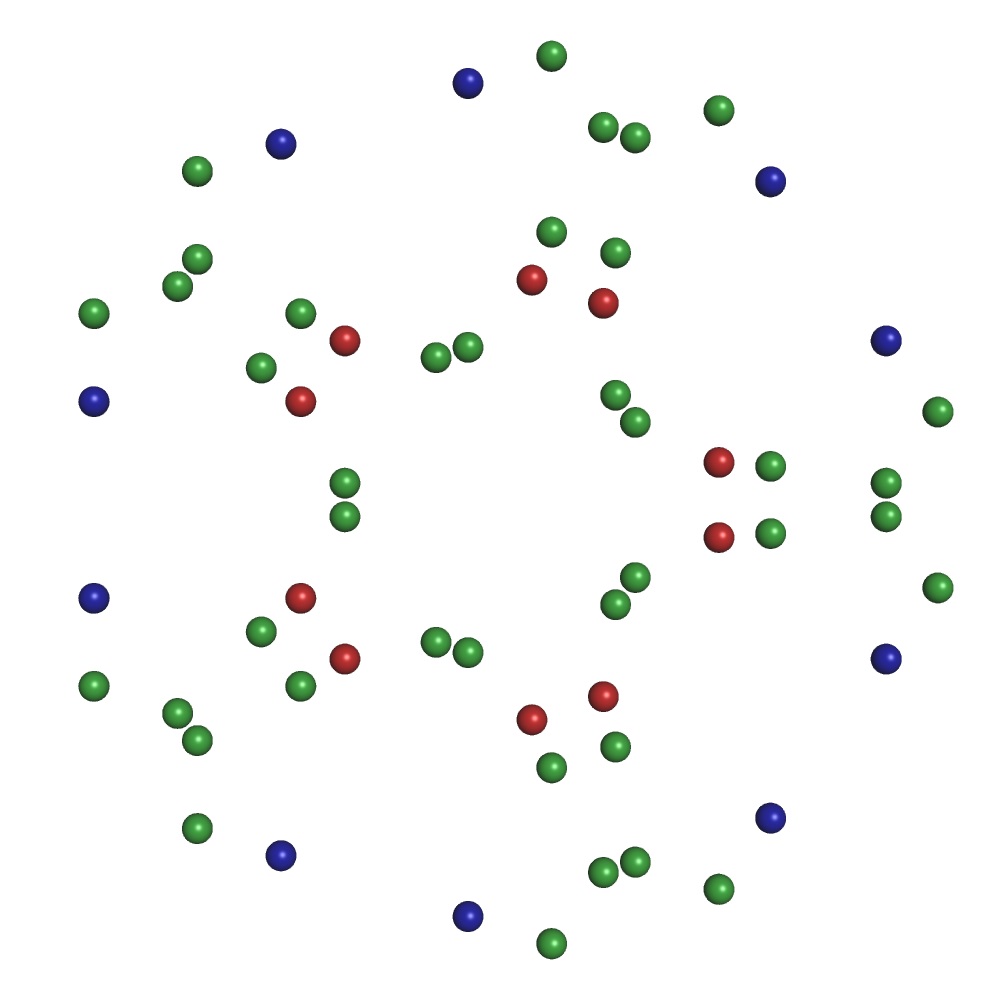} 
(c)
\end{minipage}
\caption{Examples of point sets with five-fold symetry via projection of orbits of points from the $A_4$ root lattice in $\R^4$. The lattice point $v = (1,2,4,3)$ (whose coordinates are written with respect to the basis of simple roots) is taken as seed point from the orbits under the groups (cf. \eqref{five_fold_chain}): (a) $ \h \simeq \D_{10}$, (b) $ K \simeq \text{Hol}(\Z_5)$ and (c) $\Lambda(A_4) \simeq S_5$.} 
\label{quasicrystal_fragments}
\end{figure}

\section{Nested point sets with icosahedral symmetry}\label{ico_section}

As mentioned in the Introduction, icosahedral symmetry plays a fundamental role in virology, carbon chemistry and quasicrystals. For applications in the natural sciences, it is important to distinguish between \emph{chiral} and \emph{achiral} symmetry. Chiral icosahedral symmetry is described by the icosahedral group $\I$, which consists of all the rotations that leave an icosahedron invariant and admits the presentation:
\begin{equation*}
\I = \langle g_2, g_3 : g_2^2 = g_3^3 = (g_2g_3)^ 5 = e \rangle,
\end{equation*}
where $g_2$ and $g_3$ are a two- and three-fold rotation, respectively. It has order $60$ and it is isomorphic to the alternating group $\mathfrak{A}_5$. On the other hand, achiral icosahedral symmetry corresponds to the full symmetries of an icosahedron (i.e. reflections included), and it is described by the Coxeter group $H_3$, whose order is $120$ and it is isomorphic to $\I \times \Z_2$. 

For applications in virology, we focus firstly on chiral icosahedral symmetry, since not all viral capsid are invariant under reflections. Since the icosahedral goup contains five-fold symmetry, it is not crystallographic in 3D. Its minimal crystallographic dimension, in the sense of Definition \ref{cryst_def},  is six \cite{senechal}. In particular, there are exactly three Bravais lattices left invariant by $\I$ in $\R^6$, namely the simple cubic (SC), face-centered cubic (FCC) and body-centered cubic (BCC) lattices \cite{levitov}. The point group of these lattices is the hyperoctahedral group in six dimensions, which we denote by $B_6$ (cf. \eqref{point_group}):
\begin{equation*}
B_6 = \{ Q \in O(6) : Q = M \in GL(6,\Z) \} = O(6) \cap GL(6,\Z),
\end{equation*}
which consists of all the $6 \times 6$ orthogonal and integral matrices. It is isomorphic to the wreath product $\Z^2 \wr \S_6$, where $S_6$ denotes the symmetric group on $6$ elements (see Appendix). Its order is $2^6 6! = 46,080$.  In what follows, we will focus on the SC lattice:
\begin{equation*}
\La_{SC} := \bigoplus_{i=1}^6 \Z \e_i,
\end{equation*}
where $e_i$, $i=1,\ldots, 6$ is the standard basis of $\R^6$;  its point and lattice groups coincide (cf. \eqref{point_group} and cf. \eqref{lattice_group}). The crystallographic representations of $\I$ have been classified in \cite{zappa}. They are all conjugated in $B_6$; a representative $\hat{\I}$ can be chosen as the following:

\begin{equation}\label{gen}
\hat{\I}(g_2) =  \left( \begin{array}{cccccc}
0 & 0 & 0 & 0 & 0 & 1 \\
0 & 0 & 0 & 0 & 1 & 0 \\
0 & 0 & -1 & 0 & 0 & 0 \\
0 & 0 & 0 & -1 & 0 & 0 \\
0 & 1 & 0 & 0 & 0 & 0 \\
1 & 0 & 0 & 0 & 0 & 0 
\end{array} \right), \quad \hat{\I}(g_3) = \left( \begin{array}{cccccc}
0 & 0 & 0 & 0 & 0 & 1 \\
0 & 0 & 0 & 1 & 0 & 0 \\
0 & -1 & 0 & 0 & 0 & 0\\
0 & 0 & -1 & 0 & 0 & 0 \\
1 & 0 & 0 & 0 & 0 & 0 \\
0 & 0 & 0 & 0 & 1 & 0 
\end{array} \right).
\end{equation}

$\hat{\I}$ leaves two three-dimensional subspaces invariant, denoted by $E^{\parallel}$ and $E^{\perp}$, that are both totally irrational with respect to the SC lattice. It follows that $\hat{\I}$ decomposes, in $GL(6,\R)$, into two three-dimensional irreps, usually denoted by $T_1$ and $T_2$. An explicit form of $T_1$, useful for computations, is given by 

\begin{equation}\label{ico_irrep}
T_1(g_2) = \frac{1}{2} \left( \begin{array}{ccc}
-\tau' & 1 & \tau \\
1 & -\tau & -\tau' \\
\tau & -\tau' & -1
\end{array} \right), \qquad T_1(g_3)= \frac{1}{2} \left( \begin{array}{ccc}
\tau & -\tau' & 1 \\
\tau' & -1 & \tau \\
1 & -\tau & \tau'
\end{array} \right).
\end{equation}

The projection $\pi^{\parallel} : \R^6 \rightarrow E^{\parallel}$ is given by
\begin{equation}\label{proj_H3}
\pi^{\parallel} =  \frac{1}{\sqrt{2(2+\tau)}} \left( \begin{array}{cccccc}
\tau & 0 & -1 & 0 & \tau & 1 \\
1 & \tau & 0 & -\tau & -1 & 0 \\
0 & 1 & \tau & 1 & 0 & \tau 
\end{array} \right).
\end{equation}

For achiral icosahedral symmetry, the crystallographic representations of $H_3$ are easily computed using the direct product structure $\I \times Z_2$. Specifically, if $\Gamma = \{1, -1 \}$ is the non-trivial irrep of $\Z_2$, then the representation $\hat{\I} \otimes \Gamma$ is a crystallographic representation of $H_3$, and is such that $\hat{\I} \otimes \Gamma \simeq T_1 \otimes \Gamma \oplus T_2 \otimes \Gamma$ in $GL(6,\R)$. We point out that there exist other crystallographic groups in six dimensions which contain the icosahedral group as a subgroup, and these can be found using the \texttt{GAP} package \texttt{CARAT} \cite{carat}. However, the representations of $\I$ induced by this embedding do not split into two three-dimensional irreps of $\I$, according to the classification provided by \cite{levitov}, and hence they are not suitable for the construction of nested point sets by projection presented here.   
 
In order to construct nested structures with icosahedral symmetry, we consider the set  of all the $\hat{\I}$-containing subgroups of $B_6$ (cf. \eqref{subgroups}):
\begin{equation}\label{ico_set}
\mathcal{A}_{\hat{\I}} := \{ K < B_6 : \hat{\I} < K \}.
\end{equation}

With the help of \texttt{GAP}, it is possible to compute the set $\mathcal{A}_{\hat{\I}}$. In order to  make computations efficient, we use some results from group theory. In particular, we recall that, if $G$ is a soluble group, then every subgroup of $G$ is soluble \cite{humpreys}. Since the icosahedral group is isomorphic to $\mathfrak{A}_5$, it is not soluble. Therefore, any subgroup $H$ of $B_6$ containing $\hat{\I}$ as a subgroup must not be soluble. Moreover, it cannot be Abelian (since $\I$ is not) and the order of $H$ must be divisible by $|\I| = 60$, as a consequence of Lagrange's Theorem. With these considerations, we provide the following algorithm.

\begin{alg}\label{alg_ico}
In order to determine $\mathcal{A}_{\hat{\I}}$, perform the following steps:
\begin{enumerate}
\item Compute the conjugacy classes $\mathcal{C}_i$ of the subgroups of $B_6$.
\item List a representative $K_i$ for each class $\mathcal{C}_i$.
\item Rule out those representatives which have one of the following properties:
\begin{itemize}
\item $K_i$ is soluble;
\item $K_i$ is Abelian;
\item $60 \nmid |K_i|$.
\end{itemize}
\item For each $K_i$ not ruled out, compute all the element $G_i \in \mathcal{C}_i$. If  $\hat{\I} < G_i$, then add $G_i$ to $\mathcal{A}_{\hat{\I}}$.
\end{enumerate}
\end{alg}

The algorithm was implemented in \texttt{GAP}, and the results are given in Table \ref{sottogruppi}. There are $13$ elements in $\mathcal{A}_{\hat{\I}}$, which we denote by $G_i$, for $i = 1, \ldots, 13$. A set of generators for each group $G_i$ is given in the Appendix. Clearly, $G_1$ and $G_{13}$ are the icosahedral and hyperoctahedral group, respectively, whereas $G_2$ is isomorphic to $H_3$. In Figure \ref{sottogruppi_lattice} we show the graph of inclusions of the groups $G_i$. 

\begin{table}[!t]
\begin{center}
\begin{tabular}{|ccc|}
\hline
Subgroup & Order & Index  \\
\hline
$G_1 \simeq \I$ & $60$ & 1\\
$G_2 \simeq H_3$ & $120$ & 2\\
$G_3$ & $240$ & 4 \\
$G_4$ & $1,920$ & 32\\
$G_5$ & $3,840$ & 64\\
$G_6$ & $3,840$ & 64\\
$G_7$ & $3,840$ & 64\\
$G_8$ & $7,680$ & 128\\
$G_9$ & $11,520$ & 192\\
$G_{10}$ & $23,040$ & 384\\
$G_{11}$ & $23,040$ & 384\\
$G_{12}$ & $23,040$ & 384\\
$G_{13} = B_6$ & $46,080$ & 768\\
\hline
\end{tabular}
\end{center}
\caption{Classification of the subgroups of $B_6$ containing the icosahedral group $\I$ as a subgroup.}
\label{sottogruppi}
\end{table}

\begin{figure}[!t]
\includegraphics[scale = 0.6]{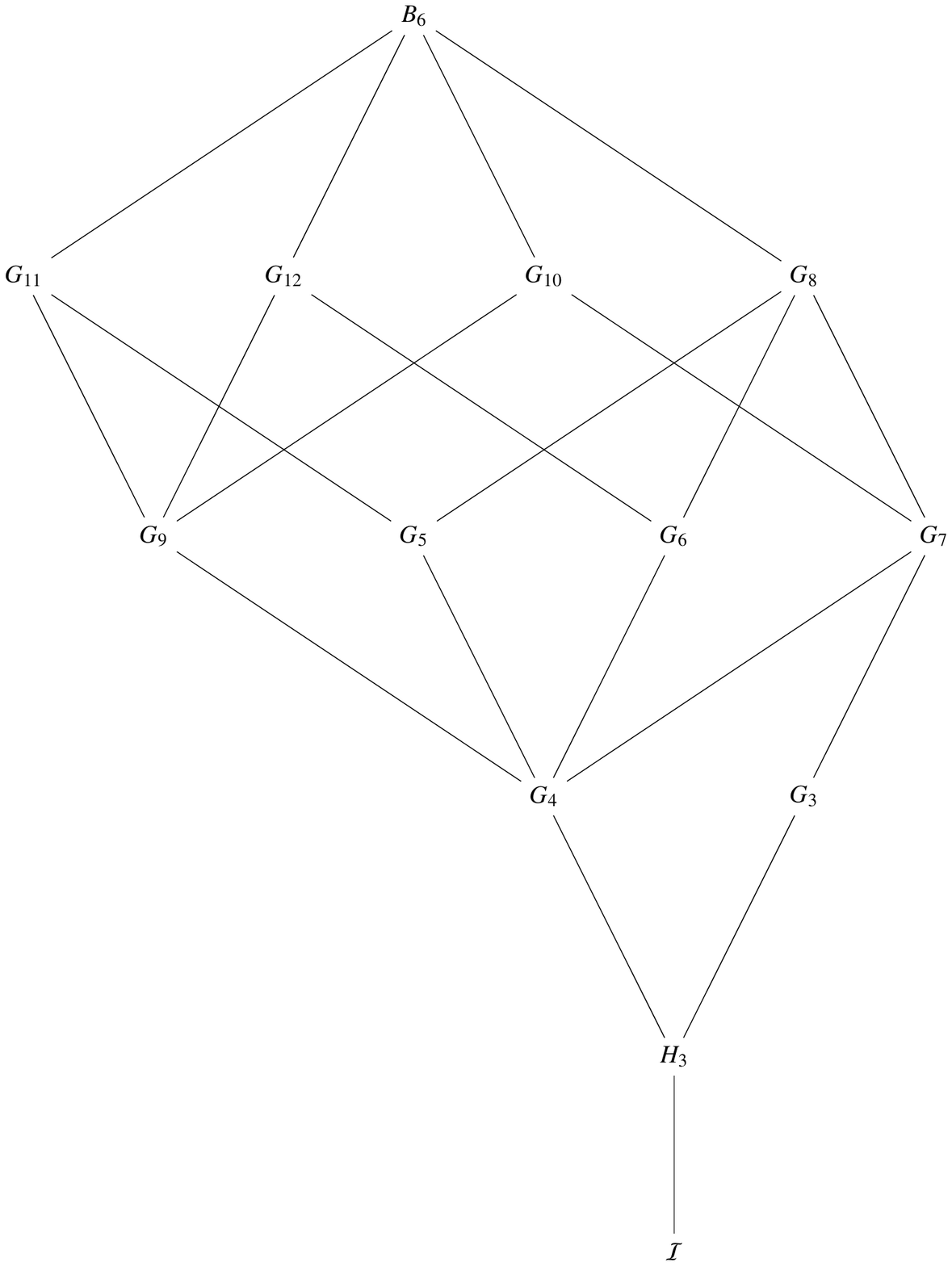}
\caption{Graph of inclusions of the subgroups containing the icosahedral group embedded into the hyperoctahedral group.}
\label{sottogruppi_lattice}
\end{figure} 

The projections into $E^{\parallel}$ of the orbits of lattice points under the groups $G_i$ produce nested point sets with icosahedral symmetry at each radial level. An example is given in Figure \ref{orbit}.  Every radial level corresponds to the union of cosets of $G_i$ with respect to $\hat{\I}$. It is worth pointing out that every group $G_i$, for $i > 3$, contains $H_3$ as well as $\I$ as subgroups. From a geometrical point of view, this implies that the resulting orbits in projection are all invariant under reflections, i.e. each radial lavel possesses full icosahedral symmetry $H_3$. This observation provides a sharper bound on the number of distinct radial levels in projection: in fact, this is given by $n/2$, which is the index of $H_3$ in $G_i$ (recall that $n$ is the index of $\I$ in $G_i$). 

\begin{figure}[!t]
\begin{center}
\includegraphics[scale = 0.15]{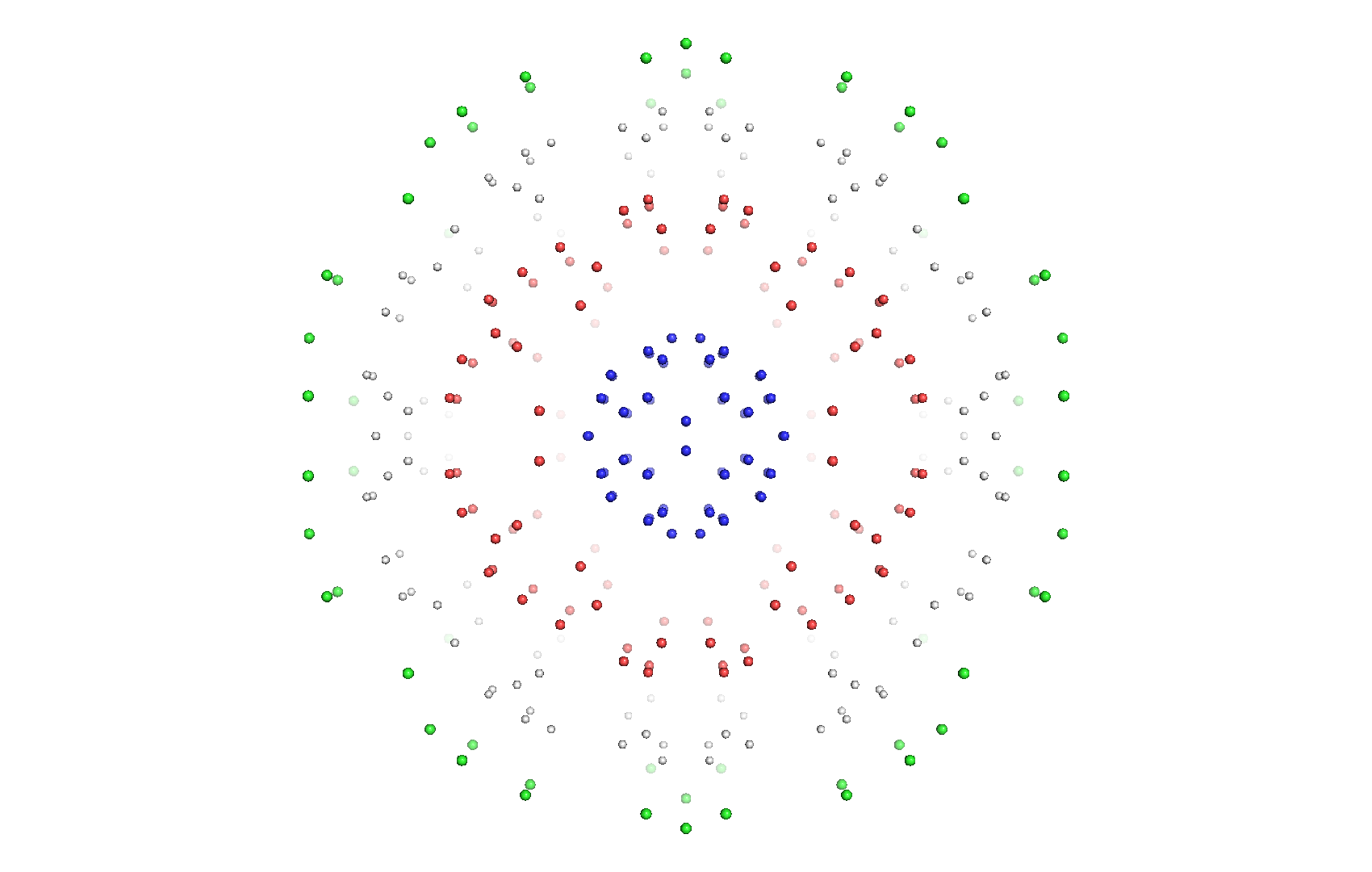}
\caption{Projected orbit of the lattice point $v = (0, 0, 1, 1, 2, 1)$ under the group $G_4$. Each layer in the resulting nested point set possesses achiral icosahedral symmetry.} 
\label{orbit}
\end{center}
\end{figure} 

\section{Applications to viral capsid architecture}\label{virus}

In this section we show that this group theoretical setup is a powerful tool to rationalise viral architecture. Specifically, the classification of the chains of subgroups of $B_6$ extending icosahedral symmetry, derived in Section \ref{ico_section}, provides a suitable mathematical framework to understand structural constraints on viral capsids. As a first step towards this goal, we identify a finite library of point arrays, corresponding to the projected orbits of $6D$ lattice points under the groups $G_i$ previously classified. Elements in this library depend on two quantities: the group $G_i \in \mathcal{A}_{\hat{\I}}$ and the lattice point $v \in \La_{SC}$. The $G_i$ are provided by our classification. As can be seen from Figure \ref{sottogruppi_lattice} and Table \ref{sottogruppi}, the first group that gives icosahedral nested shells in projection is $G_3$. The index of $G_3$ with respect to $H_3$ is $2$, therefore the number of radial levels is at most 2.  In order to obtain deeper information about the geometry of capsids, more radial levels are necessary. Therefore, we neglect the orbits of $G_3$ and consider the subgroups $G_i$, for $ i = 4, \ldots, 13$. Moreover, $v$ is chosen as follows: since the 6D lattice is infinite, we introduce a cut-off parameter $N >0$ and consider all lattice points within a six-dimensional cube: 
\begin{equation*}
I_N^6 := [-N, N] \times \ldots \times [-N,N] = [-N,N]^6 \subseteq \La_{SC},
\end{equation*}
containing $(2N+1)^6$ lattice points. In particular, we consider all orbits of the groups $G_i$ within a bounded area around the origin defined by $N$. 

Based on this set-up, the library of point arrays is obtained via the action of the group $G_i$ on the set $I_N^6$, for $i =4, \dots, 13$. This action is well-defined since $G_i$ is a subgroup of the point group of the lattice, and therefore lattice points are mapped into lattice points under elements of $G_i$. Let $D_N^{(i)} = \{v_1^{(i)}, \ldots, v_{k_i}^{(i)} \}$ be a set of distinct representatives for the orbits of $G_i$ in $I_N^6$. Since $G_4 \subseteq G_i$ for all $i = 5, \ldots, 13$, and thus their fundamental domains are contained in that of $G_4$, the set $D_N^{(4)}$ contains the sets of representatives $D_N^{(i)}$ for the groups $G_i$, $i = 5, \ldots, 13$, which are not necessarily distinct. Since we do not have information on the group $G_4$ apart from its generators, the set $D_N^{(4)}$ is computed numerically according to the following procedure:
\begin{enumerate}
\item For $v \in I_N^6$, compute $\Or_{G_4}(v)$; 
\item among all $v_i \in \Or_{G_4}(v)$ identify ${\hat v}$ with the largest number of positive components, choosing at random if two or more points fulfil this property; 
\item add ${\hat v}$ to $D_N^{(4)}$ and repeat from the start until all $v \in I_N^6$ have been considered.
\end{enumerate}
In particular, $D_N^{(4)}$ thus obtained contains $47$, $183$ and $529$ points for $N = 2,3$ and $4$, respectively. With this setup, the library is given by
\begin{equation}\label{library}
\mathcal{S}(N) := \left\{ \{\pi^{\parallel}(\Or_{G_j}(v))\} : v \in D_N^{(4)}, \; j = 4, \ldots, 13 \right\},
\end{equation} 
which by construction consists of distinct point arrays.

Once the set $\mathcal{S}(N)$ is computed for a chosen value of $N$, we retrieve the information of the viral capsid in consideration from the VIPER data bank \cite{viper}. These PDB files contain structural data of viral capsid, such as the coordinates of the atomic positions of the capsid proteins and in many cases also of the packaged genome. Following \cite{wardman}, we represent the atomic positions of the proteins by spheres of radius 1.9 $\AA$ in the visualisation tool PyMol.  In order to compare the point arrays with biological data, and hence find those point sets which best represent the capsid features, we use the following procedure:

\begin{enumerate}
\item For any group $G_i \in \mathcal{A}_{\hat{\I}}$, we compute with $\texttt{GAP}$ a transversal $T^{(i)} = (g^{(i)}_1, \ldots, g^{(i)}_{n_i})$ for the right cosets of $\hat{\I}$ in $G_i$, where $n_i$ denotes the index of $\I$ in $G_i$.
\item Given a point array $\pi^{\parallel}(\Or_{G_i}(v)) \in \mathcal{S}(N)$, we compute the set
\begin{equation*}
L^{(i)}(v) = \{ |\pi^{\parallel}(g^{(i)}_j v)| : j =1, \ldots, n_i \}. 
\end{equation*}
The cardinality of $L^{(i)}(v)$ is the number of distinct radial levels in the point set $\pi^{\parallel}(\Or_{G_i}(v))$.  We denote by $R^{(i)}_{max}(v) := \text{max}_j L^{(i)}(v)$ the largest radial level which corresponds to the outermost layer in the nesting. This is used to scale the point set so that the capsid is contained in the convex hull of the projected orbit. 
\item The rescaled orbit is then compared with the data in the PDB file. We start by selecting those point arrays whose outermost layer best represents the outermost features of the capsid. Specifically, we consider a coarse-grained representation of the capsid surface by locating the most radially distal clusters of $C_{\alpha}$ atoms using the procedure described by \cite{wardman}. Denoting these clusters by $C_k$, $k = 1, \ldots, M$, the $C_k$ can be approximated by $M$ spheres $B_k(\tilde{r})$ of radius $\tilde{r}$ (for the numerical implementation, we chose the cutoff $\tilde{r} = 10 \AA$). For any orbit $\pi^{\parallel}(\Or_{G_i}(v))$, we isolate its external point layer $L^{(\text{out})}$ by computing the points situated at distance $R_{max}$ (introduced above) from the origin. The orbit is then selected if, for every point $x \in L^{\text{(out)}}$, there exists $k \in \{1, \ldots, M\}$ such that $x \in B_k(\tilde{r})$. 
\item Among the point sets thus selected, we determine those that best match the other capsid features. For this, we isolate the inner radial levels using the decomposition of orbits into cosets and compare them with the location of the genomic material  and the inner capsid surface. The cardinalities of the point arrays are not large enough to match with atomic positions, but they rather map around material as in \cite{jess}; this comparison can be achieved via visual inspection using the surface representation of the capsid in PyMol. 
\end{enumerate}

We consider here two case studies: Pariacoto Virus and Bacteriophage MS2, both $T = 3$ capsids in the Caspar-Klug classification. These were chosen in order to facilitate comparison with \cite{jess},  where point arrays derived from affine extensions of the icosahedral group were used to generate blueprints for viral architecture, and \cite{janner4}, where virus architecture is approximated by lattices. 

\paragraph{Pariacoto Virus}
Pariacoto Virus (PaV) is a single-stranded RNA insect virus, whose X-ray crystal structure reveals approximately $35 \%$  of the RNA organised as a dodecahedral cage of duplex RNA in proximity to the  the inner capsid surface \cite{tang}. A characteristic feature of this capsid are the $60$ protrusions of approximately $15 \AA$ around the quasi three-fold axes, each formed by three interdigitated subunits. These are the outermost capsid features that we will match to the largest radial levels in the point arrays of our library in order to identify the best fit point array. For this we performed the procedure described above, and found that the best fit for this capsid is given by the projected orbit of the lattice point $\hat{v} = (2, 1, -1, -1, 0, 0)$ under the group $G_6$ (see Figure \ref{pariacoto}). This point set consists of $960$ points, arranged into $8$ radial levels. The outermost level is formed by 60 points which map onto the spikes at the $60$ local three-fold axes, see Figure \ref{pariacoto} (b). The third radial level from the origin describes the organisation of the RNA inside the capsid. This set is made up of $120$ points forming a truncated icosidodecahedron, which maps around the dodecahedral RNA cage, see Figure \ref{pariacoto} (d). The fifth radial level from the origin, located between the RNA and the spikes, consists of $120$ points, organised into 10 and 12 clusters of 6 and 5 points each, which are located around the $3$ and $5$ fold axes, respectively. In particular, we show in Figure \ref{pariacoto} (c) a close-up view of the clusters with five-fold symmetry. Note that these points provide constraints on the lengths of the protein helices and the positions of the protein subunits of type $C$. 

We point out that $G_6$ is the group of smallest order in the set $\mathcal{A}_{\hat{\I}}$ that provides a blueprint for PaV that captures the location of both capsid proteins \emph{and} the RNA collectively. The orbit of $\hat{v}$ under $G_4$ in projection, which by construction is contained in $\pi^{\parallel}(\Or_{G_6}(\hat{v})$, maps around the spikes, but totally lacks information on the organisation of the genomic material inside. Moreover, all the orbits of $\hat{v}$ under the $G_6$-containing $G_k \in \mathcal{A}_{\hat{\I}}$, i.e. $G_8$ and $G_{12}$, as well as $B_6$ (cf. Figure \ref{sottogruppi_lattice}) coincide in projection, implying that they contain no additional information on capsid architecture. Hence $G_6$ can be chosen as the six-dimensional symmetry group that induces the three-dimensional structure of the PaV capsid in projection. 

\begin{figure}[!t]
\begin{center}
\includegraphics[scale = 0.4]{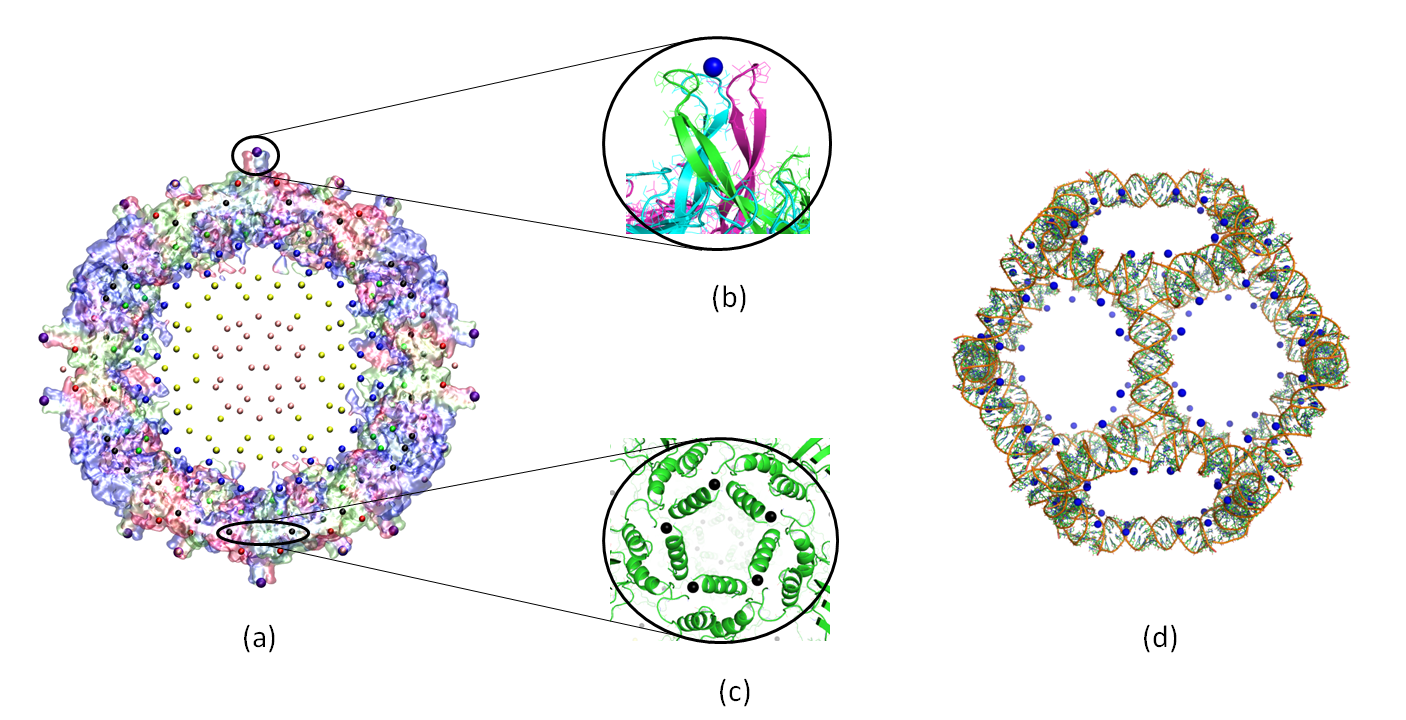}
\label{pariacoto}
\caption{Blueprints for the capsid of Pariacoto Virus (based on pdb file 1f8v). (a) Cross-section of the capsid superimposed with the projected orbit of $\hat{v} = (2,1,-1,-1,0,0)$ under the group $G_6$. The point set consists of $960$ points situated at 8 distinct radial levels which provides constraints on the capsid architecture. (b) Close-up view of the outermost layer of projected orbit, which encodes the locations of the spikes around the quasi-threefold axes. (c) The layers between the spikes and the genomic material map around the inner surface of the capsid proteins. (d) The third farthest layer from the origin gives information on RNA organisation: the 120 points, forming a truncated icosidodecahedron, map around the dodecahedral RNA cage.}
\end{center}
\end{figure}

\paragraph{Bacteriophage MS2} Like PaV, MS2 is a single-stranded RNA virus, with a $T = 3$ capsid. Cryo-electron microscopy reveals a double-shell structure in the organisation of the genomic RNA \cite{toropova}, and we will demonstrate here that our approach is able to capture this. With our procedure as above, we found that the projected orbit of $\tilde{v} = (2,1,1,-1,0,1)$ under the group $G_4$ is the point set that provides the best blueprint for the capsid (see Figure \ref{ms2}). Specifically, this orbit contains $960$ points, that are arranged, in projection, into $9$ radial levels. The two outermost layers, $L^{(9)}$ and $L^{(8)}$, map to the exterior of the capsid: $L^{(9)}$ consists of $60$ points, arranged into $12$ clusters of $5$ points each, positioned around the five-fold symmetry axes of the capsid, whereas $L^{(8)}$ has size $120$ and is made up of $20$ clusters of $6$ points, located around the three-fold axes. This is consistent with the quasi-equivalent structure of the $T = 3$ capsid. We point out that $L^{(8)}$ and $L^{(9)}$ are in fact almost situated at the same radial level (the ratio of their radii is $\simeq 1.064814$), and collectively map around the capsid exterior as demonstrated in the close-up in Figure \ref{ms2_asymmetry} (b). 

All other points of the array are from a mathematical point of view related to these outmost shells, and should therefore also map around material boundaries. We start by comparing the point array with the icosahedrally averaged cryo-EM structure at $8 \AA$ resolution in \cite{toropova}. As shown in Figure \ref{ms2} (a), the innermost radial levels of the point array define the interior of the inner RNA shell. Moreover, there are points mapping around the outer and inner surfaces of the other shell. There is a layer of points between the shells that at a first glance seems to float in space, but a close inspection of the data set reveals that they are in fact positioned around the RNA connecting the outer and inner shell, see also the close-up in Figure \ref{ms2} (c). This icosahedrally averaged data set has been obtained via a superposition of a large number of viral particles, aligned according to their symmetry axes, in order to enhance the resolution. However, in any individual particle, the RNA is organised in an asymmetric way, that is consistent with the icosahedrally averaged density. Since our point arrays are not fully constraining the structure, but are providing blueprints for the overall organisation of the virus, we expect the asymmetric organisation to be consistent with our symmetric point arrays. In order to test this hypothesis, we compare our model with the asymmetric RNA density of a cryo-EM tomogram at about 39 $\AA$ resolution (\cite{dent}, \cite{james}), see Figure \ref{ms2_asymmetry}. Since the density is shown in a cross sectional view, the density in the two shells cannot be seen in full. However, as expected, the density is consistent with the radial levels defined by the point arrays, consistent with our hypothesis that the mathematical model indeed describes material boundaries in this virus. Taken together, these results imply that the group $G_4$ is the group of smallest order in our classification that provides structural constraints on the capsid proteins and the genome organisation of MS2, and is therefore the symmetry group in 6D that describes the structure of this virus in projection. 

\begin{figure}[!t]
\begin{center}
\includegraphics[scale = 0.4]{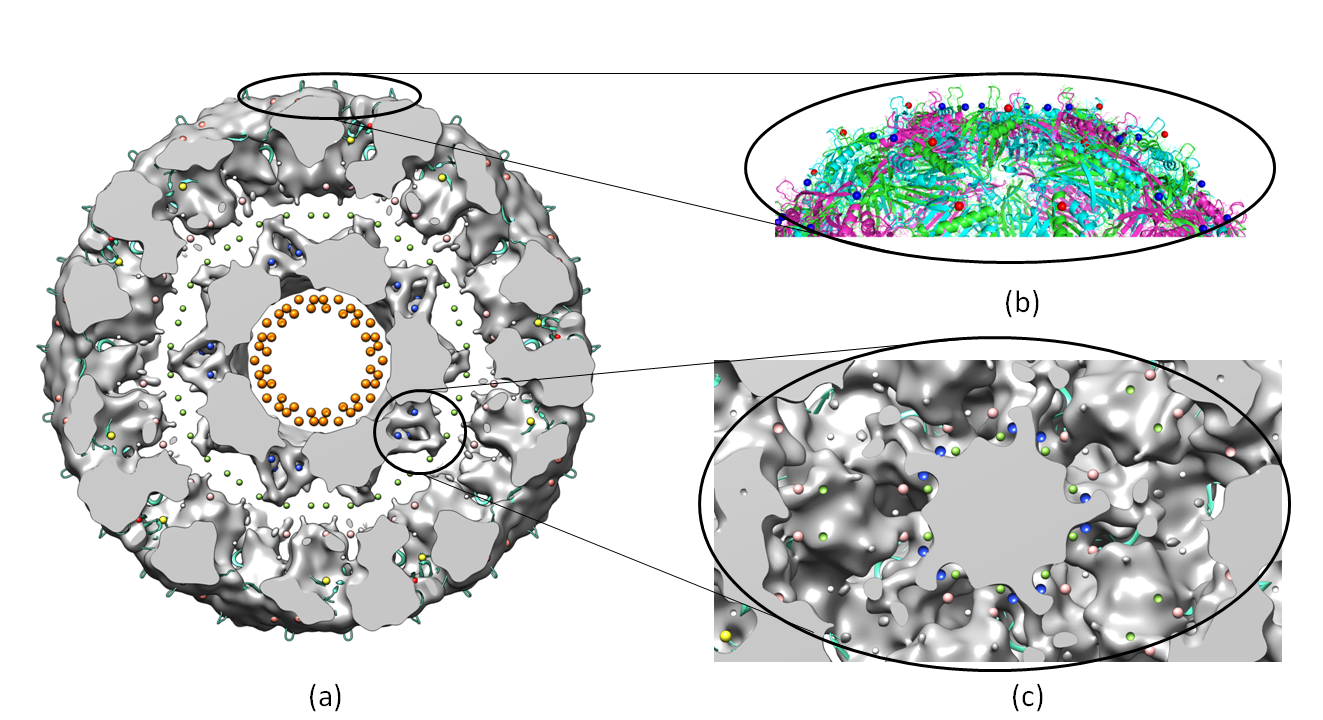}
\label{ms2}
\caption{The projected orbits of $\tilde{v} = (2,1,1,-1,0,1)$ under the group $G_4$ provides blueprints for the capsid of Bacteriophage MS2 (based on pdb file 1aq3). (a) Cross section of the capsid: the point set consists of $9$ different radial levels which encode information on the position of capsid proteins and the genomic material of the virus. (b) Close-up view of the outermost layers of the projected orbit which map around the capsid proteins. (c) Close-up view of the RNA density. The second and third innermost layers (in blue and green, respectively) map around the five-fold symmetry axes and connect the two RNA shells.}
\end{center}
\end{figure}

\begin{figure}[!t]
\begin{center}
\includegraphics[scale = 0.12]{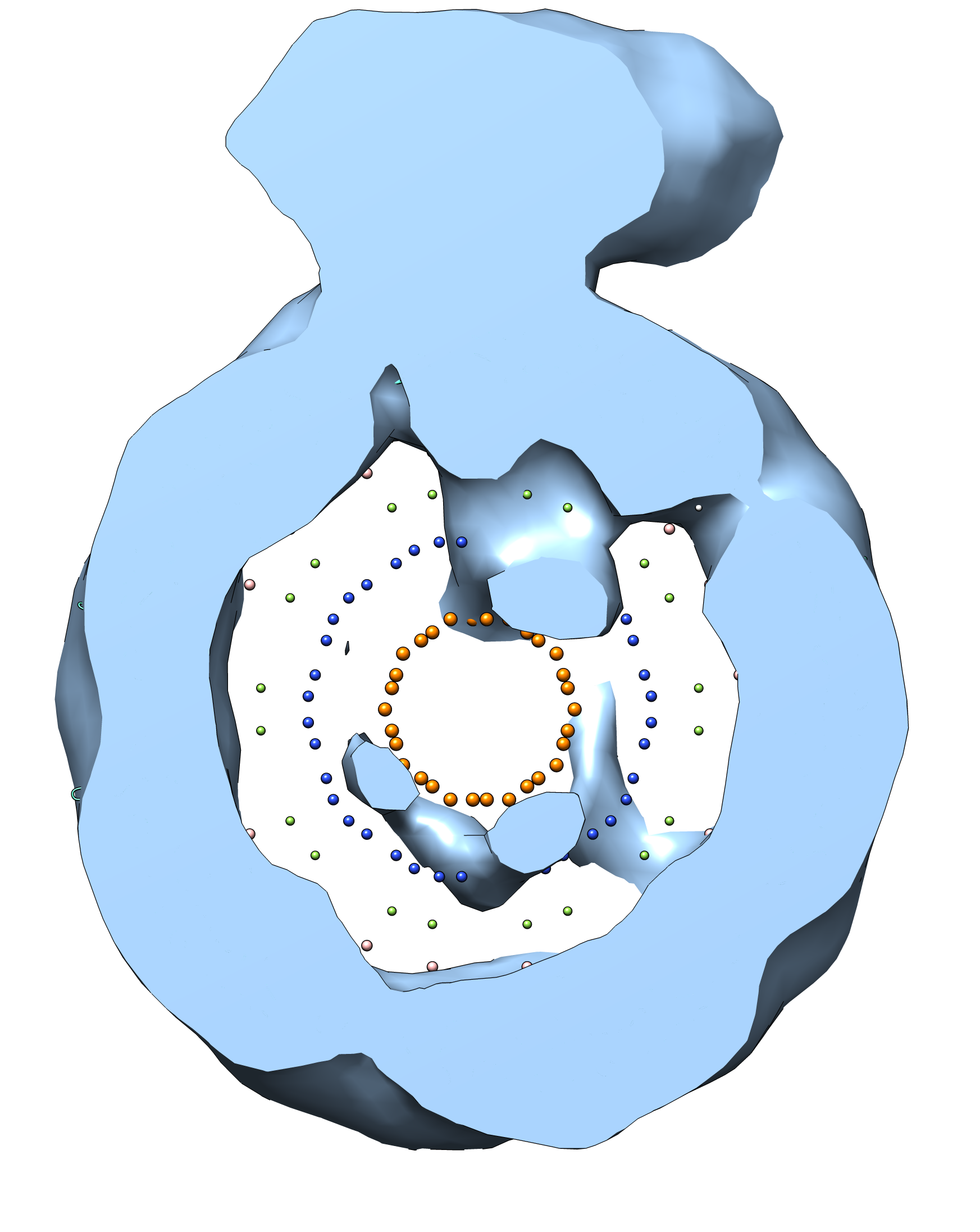}
\label{ms2_asymmetry}
\caption{A cryo-tomogram of bacteriophage MS2, adapted from \cite{dent}, superimposed on the best fit point array for bacteriophage MS2. The top of the figure shows a portion of the bacterial pilus, the natural receptor of this virus. The surface representation shows both the capsid on the exterior and the genomic RNA. The inner and outer RNA shells follow the blueprints of the array points, but realise it in an asymmetric way as expected.}
\end{center}
\end{figure}

\section{Conclusion}\label{conclusion}

The method presented here is a new way of constructing finite nested point sets with non-crystallographic symmetry from group theoretical principles. It complements previous studies, in which such point sets were constructed via affine extensions of non-crystallographic groups. The latter,  being based on the theory of infinite dimensional Kac-Moody algebras, produced infinite point sets, and a cutoff parameter had to be introduced in order to obtain finite structures. This implies that the point sets do not correspond to orbits of finite groups. The method developed in this paper, on the other hand, provides for the first time a characterisation of non-crystallographic finite multishell structures which is entirely based on the theory of finite groups in a higher dimensional space which admits a crystallographic embedding of the non-crystallographic symmetry. 

With application to viruses in mind, we discussed the case of icosahedral symmetry in detail and  provided a classification of all the subgroup chains of the hyperoctahedral group that contain a crystallographic embedding of the icosahedral group. We showed that the point sets induced by orbits of  lattice points under such groups via projection into a three dimensional invariant subspace provide a library of structural constraints on the structural organisation of viruses. In particular, we presented two case studies, Pariacoto Virus and Bacteriophage MS2, both $T = 3$ viruses, and showed that the corresponding constraint sets indicate material boundaries in these viruses. We note that also previous approaches provided good approximation for the material boundaries (\cite{janner4}, \cite{jess}, \cite{david}); however, in contrast to these approaches, we approximate here virus architecture via point arrays that are generically finite, because they stem from orbits of finite groups.  As already pointed out, the point sets display achiral icosahedral symmetry at each radial level, and hence they are invariant under reflections, i.e. the non-crystallographic Coxeter group $H_3$. However, since the point arrays only provide constraints on the structural organisation of a virus, but do not fully determine its structure, this does not imply that the virus must have full $H_3$ symmetry. Indeed, as we have discussed above, viruses may realise the blueprints given by the point arrays in an asymmetric way. This can occur, e.g., via asymmetric components in the viral capsid such as the one copy of a maturation protein that is believed to replace one of the protein dimers in the capsid shell of bacteriophage MS2 \cite{dent}, or the way in which genomic material realises the polyhedral genome organisation observed via cryo-EM. As an example for the latter, we discussed a cryo-EM tomogram of the packaged RNA of Bacteriophage MS2. However, knowledge of the possible blueprints is important, as it can be used, in combination with other techniques, in the analysis of low-resolution data of the genome organisation in viruses \cite{james}.

Viruses are known to exhibit icosahedral symmetry in their capsids due to the principle of genetic economy: the use of symmetry in the capsid organisation enables viruses to code for a small number of different types of building blocks, thus minimizing genome length, whilst building containers with a maximal number of repeats (corresponding to the order of the symmetry group) of the basic building blocks, thus achieving maximal container volume. The high level of symmetry that is observed at different radial levels, including genome organisation, may seem surprising. However, the fact that the interaction sites between genomic RNA and capsid proteins are at symmetry related positions with reference to capsid architecture may provide an explanation for the correlation between capsid architecture and genome organisation in terms of local interactions.

Moreover, our analysis of the group theoretical underpinnings of viral architecture has implications for our understanding of the dynamic properties of viruses. For example, it provides a framework for the analysis of conformational changes in viral capsid, which are structural rearrangements of the capsid proteins that are important for larger classes of viruses to become infective. Specifically, such structural transitions can be modeled with a crystallographic approach, using a generalisation of the concept of Bain strain for multi-dimensional lattices \cite{giuliana} or in the framework of the Ginzburg-Landau theory of phase transitions \cite{zappa2}. Our work opens up a new avenue for a description of such structural transitions in terms of Hamiltonians that are formulated in terms of the six-dimensional symmetry groups that induce the three-dimensional structures of the virus in projection.

More generally, previous mathematical work on affine extensions of non-crystallographic symmetries has resulted in applications beyond the area of virology for which these concepts had originally been introduced. For example, the organisation of different fullerene shells of carbon onions has been modelled with previous approaches (\cite{jannerB}, \cite{dechant1}), and we expect that our new approach should be relevant in this context as well. Moreover, a mathematical formulation of systems with non-crystallographic symmetries is a challenge in wider areas of physics, such as integrable systems, where models in terms of non-crystallographic root systems have been introduced (\cite{fring1}, \cite{fring2}); we expect that the use of projections of the higher dimensional symmetry groups, that contain non-crystallographic symmetries as crystallographic embeddings, could provide a new perspective also in this context.

\section{Appendix}

In this Appendix we provide the generators of the subgroups $G_i$, for $i =1, \ldots, 13$, which constitute the set $\mathcal{A}_{\hat{I}}$ (cf. Section \ref{ico_section}). For the computations in \texttt{GAP}, it is convenient to work with permutation instead of matrix representations. In particular, the hyperoctahedral group is isomorphic to a subgroup of the symmetric group $S_{12}$. We briefly recall this result here; for more details, we refer to \cite{baake}. 

Let $a \in \Z_2^6$ and let $\pi$ be a permutation in the symmetric group $S_6$. The set $\{ (a,\pi) : a \in \Z_2^6, \pi \in S_6 \}$ is a group under the multiplication
\begin{equation*}
(a, \pi) (b, \sigma) := (a_{\sigma} +_2 b, \pi \sigma), \qquad (a_{\sigma})_k := a_{\sigma(k)}, \; k = 1, \ldots, 6,
\end{equation*}
known as the wreath product of $\Z_2$ and $S_6$ and denoted by $\Z_2 \wr S_6$. It is isomorphic to $B_6$ via the function $T : \Z_2 \wr S_6 \rightarrow B_6$ given by
\begin{equation*}
[T(a,\pi)]_{ij} = (-1)^{a_j}\delta_{i,\pi(j)}, \quad i, j = 1, \ldots, 6.
\end{equation*}
The function $\phi : \Z_2 \wr S_6 \rightarrow S_{12}$ given by
\begin{equation*}
\phi(a, \pi)(k):= \left\{
\begin{aligned}
&\pi(k)+6a_k \quad \text{if} \; 1\leq k \leq 6 \\
&\pi(k-6)+6(1-a_{k-6}) \quad \text{if} \; 7 \leq k \leq 12
\end{aligned}
\right.
\end{equation*} 
is an injective homomorphism; the composition $\phi \circ T^{-1} : B_6 \rightarrow S_{12}$ can be used to map $B_6$ into a subgroup of $S_{12}$. In particular, the generators of $B_6$ are given by
\begin{equation*}
B_6 \simeq  \langle (1,2)(7,8), (1,2,3,4,5,6)(7,8,9,10,11,12),(6,12) \rangle,
\end{equation*}
and for the representation $\hat{\I}$ of $\I$ we have (cf. \eqref{gen}):
\begin{equation*}
\hat{\I} \simeq  \langle (1,6)(2,5)(3,9)(4,10)(7,12)(8,11), (1,5,6)(2,9,4)(7,11,12)(3,10,8) \rangle. 
\end{equation*}
 
With these results, Algorithm \ref{alg_ico} was implemented in \texttt{GAP}. The generators for the groups $G_i$, for $i = 2, \ldots, 12$, are the following: 
\begin{align*}
G_2  = & \langle  (1,6)(2,5)(3,9)(4,10)(7,12)(8,11), (1,5,6)(2,9,4)(7,11,12)(3,10,8), \\
& (1,7)(2,8)(3,9)(4,10)(5,11)(6,12)\rangle, \\
G_3  = & \langle (3,11)(4,12)(5,9)(6,10), (2,3,5,4)(6,12)(8,9,11,10),(1,2)(3,5)(7,8)(9,11) \rangle, \\
G_4 = & \langle (1,3)(2,8)(4,5,10,11)(7,9), (1,3,4,7,9,10)(2,5,12,8,11,6) \rangle, \\
G_5 = & \langle (1,8,9,7,2,3)(4,6,5)(10,12,11), (1,2)(3,5)(7,8)(9,11), (4,10) \rangle, \\
G_6 = & \langle (3,9)(6,12), (3,4,5,6)(9,10,11,12), (1,7)(6,12), (1,2,9,10,11,7,8,3,4,5)(6,12) \rangle, \\
G_7 = & \langle (1,7)(6,12), (2,8)(6,12), (1,2,9,10,11,7,8,3,4,5)(6,12), (3,4,5,12,9,10,11,6) \rangle, \\
G_8 = & \langle (1,8,9,7,2,3)(4,6,5)(10,12,11), (1,2)(3,5)(7,8)(9,11), (3,4,5,6)(9,10,11,12), (4,10) \rangle, \\
G_9 = & \langle (2,8)(6,12), (1,7)(2,5,3)(6,12)(8,11,9), (1,3,7,9)(2,12,8,6), \\
     & (1,3,2,7,9,8)(4,5,12,10,11,6) \rangle, \\
G_{10} = & \langle (1,2,6,4,3)(7,8,12,10,9), (5,11)(6,12), (1,2,6,5,3)(7,8,12,11,9), (5,12,11,6) \rangle, \\
G_{11} = & \langle (1,8,9,7,2,3), (1,7)(2,3,4)(8,9,10), (1,7)(2,3,5)(8,9,11), \\
& (2,6,3,5,4)(8,12,9,11,10), (5,11) \rangle, \\
G_{12} = & \langle (2,8)(6,12), (1,2,6,5,3)(7,8,12,11,9), (5,6)(11,12), (1,2,6,4,3)(7,8,12,10,9) \rangle. \\
\end{align*}

\paragraph{Acknowledgements.} We thank Briony Thomas, Eric C. Dykeman, Jessica Wardman and Pierre-Philippe Dechant for useful discussions, and Richard J. Bingham and James Geraets for helping in the visualisation of the cross sections of the capsids. MV thanks the York Centre for Complex Systems Analysis (YCCSA) for funding the summer project "The Art of Complexity", where part of the research for this work has been carried out. RT gratefully acknowledges a Royal Society Leverhulme Trust Senior Research fellowship (LT130088).

\bibliographystyle{unsrt}
\bibliography{orbits_viruses_bibliography}

\end{document}